\documentclass{article}

\usepackage{amsmath,amsthm,amssymb}
\usepackage{graphicx}
\usepackage{here}
\usepackage{subcaption}
\usepackage[utf8]{inputenc} % allow utf-8 input
\usepackage{url}            % simple URL typesetting
\usepackage{booktabs}       % professional-quality tables
\usepackage{amsfonts}       % blackboard math symbols
\usepackage{nicefrac}       % compact symbols for 1/2, etc.
\usepackage{microtype}      % microtypography
\usepackage{selectp}
\usepackage[round]{natbib}
\usepackage{verbatim}
\usepackage{tabu}

\usepackage{algorithm}
\usepackage{algorithmic}
\renewcommand{\algorithmicrequire}{\textbf{Input:}}
\renewcommand{\algorithmicensure}{\textbf{Output:}}
\usepackage{amsmath}
\usepackage{amssymb}
\usepackage{amsthm}
\usepackage{amsfonts} % for mathcal
\usepackage{bm}%for bold math stuff
\usepackage{bbm}
\usepackage{setspace}
\usepackage{xspace}
\usepackage{xcolor} %for additional colors
\usepackage{wrapfig}
\usepackage{array} %For table stuff like m
\usepackage{multirow}
\usepackage{enumitem}
\usepackage{stmaryrd}
\usepackage{mdframed}
\usepackage{thm-restate}

\usepackage{sty/notations}

\theoremstyle{plain}
\newtheorem{theorem}{Theorem}[section]
\newtheorem{proposition}[theorem]{Proposition}
\newtheorem{lemma}[theorem]{Lemma}

\theoremstyle{definition}
\newtheorem{definition}[theorem]{Definition}

\theoremstyle{remark}
\newtheorem{remark}[theorem]{Remark}

\usepackage[disable]{todonotes}

\usepackage[preprint]{sty/neurips_2023}

\usepackage[utf8]{inputenc} % allow utf-8 input
\usepackage[T1]{fontenc}    % use 8-bit T1 fonts
\usepackage{hyperref}       % hyperlinks
\usepackage{url}            % simple URL typesetting
\usepackage{booktabs}       % professional-quality tables
\usepackage{amsfonts}       % blackboard math symbols
\usepackage{nicefrac}       % compact symbols for 1/2, etc.
\usepackage{microtype}      % microtypography
\usepackage{xcolor}         % colors

\title{Local and adaptive mirror descents in extensive-form games}

\author{%
  C\^ome Fiegel\thanks{\texttt{come.fiegel@normalesup.org}}\\
  CREST, ENSAE, IP Paris, Paris, France \\
  \And
  Pierre M\'enard \\
    ENS Lyon, Lyon, France \\
  \And
  Tadashi Kozuno\\
  Omron Sinic X, Tokyo, Japan\\
  \And
  R\'emi Munos \\
  Deepmind, Paris, France \\
  \And
  Vianney Perchet \\
  CRITEO AI Lab, Paris, France\\
  \And
  Michal Valko \\
  Deepmind, Paris, France\\
}

\begin{document}

\maketitle

\begin{abstract}

We study how to learn $\epsilon$-optimal strategies in zero-sum imperfect information games (IIG) with \textit{trajectory feedback}. In this setting, players update their policies sequentially based on their observations over a fixed number of episodes, denoted by $T$. Existing procedures suffer from high variance due to the use of importance sampling over sequences of actions \citep{steinberger2020dream, mcaleer2022escher}. To reduce this variance, we consider a \textit{fixed sampling} approach, where players still update their policies over time, but with observations obtained through a given fixed sampling policy. Our approach is based on an adaptive Online Mirror Descent (OMD) algorithm that applies OMD locally to each information set, using individually decreasing learning rates and a \textit{regularized loss}. We show that this approach guarantees a convergence rate of $\tilde{\mathcal{O}}(T^{-1/2})$ with high probability and has a near-optimal dependence on the game parameters when applied with the best theoretical choices of learning rates and sampling policies. To achieve these results, we generalize the notion of OMD stabilization, allowing for time-varying regularization with convex increments.

% We study how to learn $\epsilon$-optimal strategies in zero-sum imperfect information games (IIG) with a \textit{trajectory feedback}, in which players sequentially update their policies using only their observations over $T$ episodes. Such procedures usually suffer from a high variance in large games because of the low probability of choosing each sequence of action that appears in the importance sampling term. One way of reducing this variance is \textit{fixed sampling}: each player still updates its policy over time, but its observations are obtained through the use of a given fixed policy instead of the current one. We show that a specific adaptive Online Mirror Descent (OMD) approach in this setting can be interpreted as applying an OMD procedure locally on each information set, using individually decreasing learning rates and a \textit{regularized loss}. We show that this approach guarantees a $\tcO(T^{-1/2})$ convergence rate with high probability, and enjoys a near-optimal dependence with respect to the game parameters when applied with the best theoretical choices of learning rates and fixed policies. These results are obtained through a generalization of the notion of OMD stabilization that we propose, allowing for a time-varying regularization with convex increments. 
\end{abstract}

\section{Introduction}

The extensive-form representation of a game \citep{osborne1994course} can be depicted as a tree whose nodes correspond to the game states. At each state, the players choose some available actions and, based on these choices, the game transitions to the next state among the current state's children.

In imperfect information games (IIGs), players may only have access to partial information about the current game state upon taking action. Therefore, the state space is partitioned for each player into multiple information sets, which consist of indistinguishable states from the player's perspective. With perfect recall \citep{kuhn1950extensive}, when players remember their previous moves, each space of information sets also has a tree structure. 

 We focus more specifically on zero-sum IIGs represented in an extensive form under the perfect recall assumption, where the gains of one player, conventionally called the max-player, are equal to the losses of his opponent, the min-player. The primary goal is to design an algorithm learning $\epsilon$-optimal strategies \citep{neumann1928zur}. To achieve this, one can use the self-play framework, where an agent controls both players for $T$ episodes. At the beginning of each episode, the agent prescribes a strategy for each player. The agent then observes the results and updates the players' strategies for the next episode based on the outcome of the game. After $T$ episodes, this protocol returns a guess of strategies with a small exploitability gap \citep{ponsen2011computing}. In this learning framework, the agent has very limited feedback, only observing the rewards along each trajectory, as opposed to richer feedback that would for example include all possible rewards and all transition probabilities, \citep{zinkevich2007regret,hoda2010smoothing,tammelin2014solving,kroer2015faster,burch2019revisiting} %, such as observing a complete traversal  of the
 %\todom{observing a complete traversal of the -> having access to the entire?  (it otherwise also looks like we could change the tree each time...)}
unrealistic in large games.
 %\todoVP{not sure that transversal is a universal concept (I am not 100\% sure what it means for instance)}

To deal with this learning framework, a well-studied approach is to unilaterally minimize the regret of each player during the interactions with the game, i.e. the difference between the cumulative gain the player would have obtained had he played the best fixed a posteriori policy and the cumulative gain obtained by following the sequence of policies. The key observation is that by minimizing the regret of both players, the average policies over the sequence of policies generated during the process converge toward optimal strategies at the rate of order $\cO(1/\sqrt{T})$ \citep{cesa-bianchi2006prediction, kozuno2021learning}. Regret minimizers such as \CFR-based algorithm or online mirror descent (OMD) \citep{hoda2010smoothing,kroer2015faster} can be used, leading to optimal rates (with respect to the game size) with the latter option \citep{bai2022nearoptimal, fiegel2023adapting}.

Since the agent only observes trajectories of the game, an importance sampling estimate \citep{auer2003nonstochastic} of gain (or loss) is fed to the regret minimizer. However, the estimate of this loss usually suffers from high variance due to two reasons. First, the same sequence of policies is used to minimize the regret and to collect trajectory, making the players strive to fulfill two competing goals: play a policy with small regret and play a policy leading to a small variance loss. Second, importance sampling is applied to sequences of actions, that have in large games a very small probability of being played, leading to empirically large importance sampling weights and ultimately inflating the variance of the gain estimates. 

To mitigate this issue, regularization and biasing the estimates can help \citep{kozuno2021learning,bai2020near}. However, the high variance of the gain estimates remains problematic with large games, for which the algorithms are generally coupled with function approximation \citep{steinberger2020dream,mcaleer2022escher}. For instance, neural networks are particularly susceptible to noise \citep{zhang2021understanding}. A natural question is thus whether it is possible to learn optimal strategies without relying on the importance-sampling over the sequence of actions.

To this aim, we consider a particular case of the self-play framework: the fixed policy sampling framework \citep{lanctot2009monte}. In this setting, a fixed policy is used to collect the trajectories of the game. Precisely, at each round, one player, let's say the min-player, follows the fixed sampling policy to play against the current policy of the max-player. The collected trajectory is then used to update the current policy of the min-player. In the next episode, the max-player will follow a sampling policy against the current policy of the min-player, and so on. The outcome sampling \MCCFR algorithm adopts this framework to update the two players' policy by regret minimization, feeding the \CFR algorithm with gain estimated via importance sampling \citep{lanctot2009monte,bai2020near,farina2021bandit}. 

Recently, \citet{mcaleer2022escher} proposed the \ESCHER algorithm that removes the need for importance sampling in this framework. In particular, as the \CFR algorithm is invariant by re-scaling of the gains and the weights of the sampling policy are fixed, \ESCHER can directly operate with the unweighted history cumulative gain \citet{bai2020near}. Unfortunately, it still requires access to an oracle that provides this history of cumulative gains at an arbitrary information set. 

Nonetheless, the insight of \citet{mcaleer2022escher} cannot be used directly for OMD-based algorithms as they are not scale-invariant. Furthermore, the OMD-based algorithms generally work at the global game level whereas \CFR-based algorithms work at the local level of the information set \citet{bai2020near}, making local adaptation to the problem easier.

% \citep{mcaleer2022escher} proposed, in order to remove the importance-sampling term, to use a fixed sampling policy for the observations of both players. In practice

% \todoCo{add a paragraph for the fixed sampling policy, and maybe at some point (no idea where for now) talk about the issue with OMD adaptation}

\paragraph{Contributions}

We make the following main contributions:
\vspace{-0.2cm}
\begin{itemize}[itemsep=-2pt,leftmargin=6pt]

    \item We propose the \LocalOMD algorithm, in the fixed policy sampling framework, that allows adaptive learning rates and does not require importance-sampling over the sequence of actions but only for the current action. \LocalOMD is computationally efficient as it can be seen as a regret minimization procedure applied to a regularized loss, locally on each information set \citep{farina2019regret}.

    \item We prove that \LocalOMD, with an appropriate sampling policy and choice of learning rates, has a $\tcO\pa{H^3(\AX+\BY)/\epsilon^2}$\footnote{For algorithms with a probability at least $1-\delta$ of a correct output, the symbol $\tcO$ hides dependencies logarithmic in $\AX,\BY$ and $\delta$} near-optimal sample complexity for learning $\epsilon$-optimal strategies, where $H$ is the height of the tree, $\AX$ the total number of available actions for the min-player and $\BY$ the same quantity for the max-player. This sample complexity was also achieved in a similar setting by \texttt{BalancedCFR} \citep{bai2022nearoptimal}, but with a less natural procedure that updates the policy at one depth at a time.

    \item We also prove that \LocalOMD achieves a $\tcO\pa{1/\epsilon^2}$ sample complexity, ignoring the game and policy-dependent parameters, with any choice of positive fixed sampling policy.
    \item We generalize the dual-stabilization technique introduced by \citet{fang2020online} to analyze OMD with a time-varying regularization as long as the increments of the regularization are convex.

    \item We provide tabular experiments and observe that our algorithm obtains similar results as existing baselines, but requires fewer updates overall as only one of the two players' strategies is updated at each iteration.
    
    % We show that with a small tweak, which we call dual-stabilization as in \citep{fang2020online}, OMD enjoys the same classical guarantees with a time-varying regularization as long as the increments of the regularization are convex.

    % \item We provide tabular experiments and observe that our algorithm obtains similar results as existing baselines, but requires fewer updates overall, as only one of the two players' strategies is updated at each iteration.
\end{itemize}
\section{Settings and fixed sampling procedure}

\subsection{Extensive-form games and regret}

\paragraph{Game definition} We consider a finite zero-sum IIG game $(\cS,\cX,\cY,\cA,\cB,p,\ell)$ with perfect recall. Given two behavioral policies $\mu=(\mu(\cdot|x))_{x\in\cX}$ and $\nu=(\nu(\cdot|y))_{y\in\cY}$, one episode of such game proceeds as follows: An initial game state $s_1\sim p(\cdot|s_0)$ is first sampled in the set of states $\cS$ according to the transition function $p$, starting from the root $s_0$ of the tree. At depth $h$, the min- and max- players respectively observe the information set $x_h$ and $y_h$ associated with the current state $s_h$ in the spaces of information sets $\cX$ and $\cY$ (these spaces being two partitions of $\cS$), then simultaneously choose and execute actions $a_h\sim \mu(\cdot|x_h)$ and $b_h\sim\nu(\cdot|y_h)$ in the sets of legal actions $\cA(x_h)$ and $\cB(y_h)$. As a result, the state transitions to a new state $s_{h+1}\sim p(\cdot|s_h,a_h,b_h)$ in $\cS$, with  the min- and max- players getting respectively the losses $\ell_h\sim \ell(\cdot | s_h,a_h,b_h)$ in $[0,1]$ and $1-\ell_h$ according to the loss distribution $\ell$. This is repeated until a final state $s_H$ of a fixed depth $H$ is reached, after which the episode finishes.

\paragraph{Policies and actions} We will denote by $\maxpi$ and $\minpi$ the set of behavioral policies of the min- and max- players. Because of the perfect recall assumption, such policies, with an independent stochastic choice of action for each information set, are enough to describe the entire set of strategies \citep{larakibookgame}. We will also denote by $\AX$ and $\BY$ the total number of actions for respectively the min- and max- players, i.e.
\[\AX:=\sum_{x\in\cX}\,\abs{\cAx}\quad \textrm{and}\quad \BY=\sum_{y\in\cY}\,\abs{\cBy}\]

\paragraph{Regret and $\epsilon$-optimal strategies} We are interested in learning $\epsilon$-optimal policies through self-play over multiple episodes. A useful notion for this objective is the regret as explained in the introduction. We first define the value $V^{\mu,\nu}=\E^{\mu,\nu}[\sum_{h=1}^H \ell_h]$ as the expected sum of losses (for the min-player) with respect to a pair of policies $(\mu,\nu)\in\maxpi\times \minpi$. Given a sequence $(\mu^t,\nu^t)_{t\in[T]}$ in $\maxpi\times \minpi$, the regrets of the min- and max- players are then defined by
\[\regret^T_{\mathrm{min}}:= \max_{\mu^\dagger \in \maxpi} \sum_{t=1}^T \pp{V^{\mu^t, \nu^t} - V^{\mu^\dagger, \nu^t}}\quad\textrm{and}\quad
    \regret^T_{\mathrm{max}}:= \max_{\nu^\dagger \in \minpi} \sum_{t=1}^T \pp{V^{\mu^t, \nu^\dagger} - V^{\mu^t, \nu^t}}\,.\]
Minimizing the regret of both players leads to the computation of an $\epsilon$-optimal profile (equivalent to an $\epsilon$-Nash equilibrium for two players zero-sum games) through the computation of an average of the policies. The following theorem quantifies this statement under the perfect recall assumption.

\begin{theorem} \label{thm:folklore}\citep{cesa-bianchi2006prediction,kozuno2021learning}
    From a sequence $(\mu^t,\nu^t)_{t\in [T]}$ in $\maxpi\times \minpi$ define the time-averaged profile $(\overline{\mu},\overline{\nu})$, then $(\overline{\mu},\overline{\nu})$ is $\epsilon$-optimal with
    \[\varepsilon=\pa{\regret^T_{\mathrm{min}}+\regret^T_{\mathrm{max}}}/T\,.\]
\end{theorem}

It especially shows that both averaged strategies converge to the set of optimal strategies as long as the regret of both players is sub-linear.

We now focus on the min-player point of view because of the symmetry of the game. Indeed, the following ideas will apply exactly the same way to the max-player, using the losses $1-\ell_h$ instead.

\paragraph{Perfect recall and realization plan} Thanks to the perfect recall assumption, for any information set $x\in\cX$ and action $a\in\cAx$, we know the existence of a unique depth $h\in[H]$ and history $(x_1,a_1,...,x_h,a_h)$ such that $x_h=x$ and $a_h=a$. Using this unique history, we define the realization plan $\mu^{}_{1:}\in\R^\AX$ \citep{VONSTENGEL1996220} associated to a policy $\mu\in\maxpi$ with, for any $x\in\cX$ and $a\in\cAx$:
\[\mu_{1:}^{}(x,a):=\Pi_{i=1}^h\ \mu(a_i|x_i)\, .\]
It denotes the combined probability of choosing actions that lead to $(x, a)$. We will especially define $Q_\mathrm{max}:=\{\mu^{}_{1:},\mu\in\maxpi\}$ the treeplex, i.e. the set of all possible realization plans. This set is a convex polytope of $\R^\AX$ as a set of solutions of linear equalities under positivity constraints.

\paragraph{Loss and regret linearization}

For $\nu$ a max-player policy, the unique history also leads to the definition of the adversarial transitions $p_{1:}^\nu \in \R^\cX$ and adversarial losses $\ell^\nu\in\R^\AX$ with:
\[
    p_{1:}^\nu(x) := p(x_1|s_0)\prod_{i=2}^h p^\nu(x_i|x_{i-1},a_{i-1})\quad\textrm{and}\quad
    \ell^\nu(x,a):= p^\nu_{1:}(x)\ell_h^\nu(x,a)\]
where $p(x_1|s_0)$ is the probability that $x_1$ is initially observed by the min-player, and, assuming that the max-player policy is set to $\nu$, $p^\nu(\cdot |(x_{i-1},a_{i-1}))$ denotes the probability of transitioning to $x_i$ when $(x_{i-1},a_{i-1})$ is reached, and $\ell_h^\nu$ the average loss $\ell_h$ associated to $a$ when $x$ is reached. Similarly to the realization plan, the adversarial transitions denote the combined probability of both Nature and max-player actions that lead to $x$, assuming that the min-player plays the actions $(a_1,...,a_{h-1})$.

Using a chain-rule argument, we get the relation, given a pair of policies $(\mu,\nu)\in\maxpi\times\minpi$,
\[V^{\mu,\nu}=\scal{\ell^\nu}{\mu_{1:}}\,,\]
where $\scal{\cdot}{\cdot}$ is the standard inner product of $\R^{\AX}$, defined by $\scal{z_1}{z_2}:=\sum_{x\in\cX}\sum_{a\in\cAx}z_1(x,a)z_2(x,a)$. The regret can then be rewritten
\[\regret^T_{\mathrm{min}}= \max_{\mu^\dagger \in \maxpi} \sum_{t=1}^T \scal{\ell^t}{\mu_{1:}^t-\mu_{1:}^\dagger}\qquad \textrm{where} \qquad\ell^t:=\ell^{\nu^t}\, ,\]
which effectively reduces the problem to a linear regret problem over the convex polytope $Q_{\mathrm{min}}$ of realization plans.

Several techniques exist to sequentially choose policies $(\mu^t)_{t\in [T]}$ minimizing $\regret^T_{\mathrm{min}}$, assuming that the losses $\ell^t$ are observed after each round $t$ \citep{hoda2010smoothing}. However, in the \textit{trajectory feedback} setting, these losses are not observed, and can only be estimated from the observation of the trajectories $(x^t_1,a^t_1,...,x^t_H,a^t_H)$ and partial losses $(\ell^t_1,...,\ell^t_H)$ of each round.

\subsection{Fixed sampling policy}\label{sec:fixed_sampling}

In the \textit{fixed sampling} framework \citep{lanctot2009monte}, both players always use the same policy for the observations of the trajectory. However, the two observations can not be done simultaneously with such an approach, as the learning would then be quite naive. The solution, summarized in Algorithm~\ref{alg:learning}, is for the two players to take turns between an observation phase, in which they play their fixed sampling policy $\mu^s$ or $\nu^s$, and an interaction phase, in which they play their updated policy $\mu^t$ or $\nu^t$. The underlying idea is that the observation phase lets each player observe how the game unfolds against the opponent in its interaction phase, playing its updated policy. Given upper-bounds of the regrets $\regret^T_{\mathrm{min}}$ and $\regret^T_{\mathrm{max}}$ associated to the sequence $(\mu^t,\nu^t)_{t\in[T]}$, the previous Theorem~\ref{thm:folklore} then characterizes the optimality of the outputted time-averaged profile $(\overline{\mu},\overline{\nu})$.

\begin{algorithm}[t] 
\caption{Learning procedures with fixed sampling policies for two players}
\label{alg:learning}
\begin{algorithmic}[1]
                \STATE \textbf{Input:}\\
                Fixed sampling policies $\mu^s$ and $\nu^s$\\
                Initial policies $\mu^1$ and $\nu^1$ and update procedure for each player\\
                \STATE \textbf{Output:}\\
                The time-averaged policies $\overline{\mu},\overline{\nu}$ of Theorem~\ref{thm:folklore}
                \STATE \textbf{Algorithm:}\\
                For $t=1$ to $T$\\
                ~~~~ The min-player observes the full outcome of an episode with the policies $(\mu^s,\nu^t)$\\
                ~~~~ The max-player observes the full outcome of an episode with the policies $(\mu^t,\nu^s)$\\
                ~~~~ The min- and max-player respectively update $\mu^{t+1}$ and $\nu^{t+1}$ based on their past observations
\end{algorithmic}
\end{algorithm}

This framework can remove the global importance sampling term of the loss, which reduces the variance to make them more suitable beyond the tabular setting \citep{mcaleer2022escher}. Furthermore, it allows more aggressive policies, as the observation side is handled by the sampling strategy. The immediate downside is that this sampling policy must be fixed in advance, which requires defining a good sampling policy beforehand (for example the balanced policy of Remark~\ref{rmk:kappa} below).

From now on, we again focus on the min-player for the same symmetry reasons. The next paragraph characterizes the efficiency of such sampling strategy.

\paragraph{Estimated regret}

 Based on the min-player observations, we define $\hat{\regret}^T_\textrm{min}$ the estimated regret by
\[\hat{\regret}^T_\textrm{min}:=\max_{\mu^\dagger\in\maxpi}\sum_{t=1}^T\scal{\hell^t}{\,\mu_{1:}^t-\mu_{1:}^\dagger}\]
where the $\hell^t$ are the importance-sampling estimated loss vectors, defined for each information set $x$ of depth $h$ and action $a\in\cAx$ by
\[\hell^t(x,a):=\frac{\indic{x=x_h^t,a=a_h^t}}{\mu_{1:}^s(x,a)}\ell_h^t\]
with $x_h^t$ the visited information set, $a_h^t$ the chosen action and $\ell_h^t$ the associated loss at depth $h$ of episode $t$.

The following theorem states that upper-bounding this estimated regret is enough to upper-bound the actual regret, up to an additional additive term. Its proof is given in Appendix~\ref{app:approximation} and relies on Bernstein-type inequalities.

\begin{theorem}\label{thm:estimation}
    Assume that the estimated losses are obtained with a fixed positive sampling policy $\mu^s$ as above. Then, for any sequence $(\mu^t)_{t\in [T]}$ of $\maxpi$ and any $\delta\in(0,1)$, the following bound holds with a probability at least $1-\delta$
    \[\regret^T_\textrm{min}\leq \max\left\{\hat{\regret}^T_\textrm{min},0\right\}+ 4\sqrt{ \iota H\kappa(\mu^s) T}\]
    where
    \[\iota:=\log\pa{\frac{\AX+1}{\delta}}\quad\textrm{and}\quad\kappa(\mu^s):=\max_{\mu\in\maxpi}\sum_{x\in\cX}\sum_{a\in\cA_x}\frac{\mu_{1:}(x,a)}{\mu_{1:}^s(x,a)}\, .\]
\end{theorem}

\begin{remark}\label{rmk:kappa}
    The quantity $\kappa(\mu^s)$ can be efficiently computed recursively for each of the sub-trees induced by an information set $x\in\cX$, and we will denote by $\kappa(\mu^s|x)$ the associated quantities. The same recursion shows that the \textit{balanced policy} $\mu^\star$, which plays proportionally to the total number of actions of each sub-tree, minimizes all these local quantities and satisfies $\kappa(\mu^\star)=\AX$. The related computations are provided in Appendix~\ref{app:kappa}.
\end{remark}
\section{Adaptive Mirror Descent}

We shall now focus on the update procedure the min-player can use to minimize this estimated regret. Let us first define some important notions of convex optimization. 

\begin{definition}
    Let $\Omega\subset \R^n$ be a non-empty open convex, and $\overline{\Omega}$ be its closure. A function $\Psi:\overline{\Omega}\xrightarrow[]{}\R$ is said to be Legendre if $\Psi$ is strictly convex, continuously differentiable on $\Omega$ and 
    \[\forall{y\in \overline{\Omega}\backslash \Omega},\ \lim_{x\xrightarrow[]{} y} \norm{\nabla \Psi(x)}=+\infty\, .\]
    The Bregman divergence $\breg_{\Psi}: \overline{\Omega}\times \Omega\xrightarrow[]{}\R$ of a Legendre function $\Psi$ is
    \[\breg_{\Psi}(x,y):=\Psi(x)-\Psi(y)-\scal{\nabla\Psi(y)}{x-y}\, .\]
    Note that the Bregman divergence can be more generally defined for any convex function differentiable on $\Omega$, although some key properties are lost. The Fenchel conjugate $\Psi^\star:\R^n\xrightarrow[]{}\R\cup\left\{+\infty\right\}$ of $\Psi$ is defined by
    \[\Psi^\star(\xi)=\sup_{x\in\overline{\Omega}}\scal{\xi}{x}-\Psi(x)\]
\end{definition}

\subsection{OMD and dilated entropy}

In an extensive-form game with perfect recall, algorithms based on the Online Mirror Descent (OMD) typically compute at each time step $t$ the update
\begin{equation}
\label{eq:omd}\mu^{t+1}=\argmin_{\mu\in\maxpi} \scal{\hell^t}{\mu_{1:}}+\breg_\Psi(\mu^{}_{1:},\mu_{1:}^t) \tag{OMD}
\end{equation}
where $\hell^t$ is the estimated loss and $\Psi:Q_\mathrm{min}\xrightarrow[]{} \R$ a Legendre regularizer. The key step is then the choice of the regularizer.

\paragraph{Dilated entropy}

A common choice of regularizer is the dilated entropy \citep{hoda2010smoothing, kroer2015faster}. It requires for each $x\in\cX$  a Legendre regularizer $\Psi_x$ over a convex domain $\overline{\Omega_x}\subset \R^{\abs{\cAx}}_{\geq 0}$ that contains the simplex $\Delta_{\cAx}:=\left\{\mu,\ \sum_{a\in\cA(x)}\mu(a)=1 \right\}$. For a given list of positive weights $\alpha=(\alpha(x))_{x\in\cX}$, the dilated entropy $\Psi^{\mathrm{dil}}_{\alpha}$ satisfies for any $\mu\in\maxpi$:

\[\Psi^{\mathrm{dil}}_{\alpha}(\mu_{1:}):=\sum_{x\in\cX}\alpha(x)\mu_{1:}(x)\Psi_x\pa{\mu(\cdot | x)}\quad\textrm{where}\quad \mu_{1:}(x):=\sum_{a\in\cAx}\mu_{1:}(x,a)\: .\]
Using this dilated entropy as the regularizer, the OMD updates become
\[\mu^{t+1}=\argmin_{\mu\in\maxpi} \scal{\hell^t}{\mu_{1:}}+\breg^\mathrm{dil}_{\alpha}(\mu^{}_{1:},\mu_{1:}^t)\]
where $\breg^\mathrm{dil}_{\alpha}(\mu^{}_{1:},\mu_{1:}^t):=\sum_{x\in\cX}\alpha(x)\mu_{1:}(x)\breg_x(\mu^{}_{1:}(\cdot | x), \mu^{t}_{1:}(\cdot | x))$ and $(\breg_x)_{x\in\cX}$ are the individual Bregman divergences of the $(\Psi_x)_{x\in\cX}$. The benefits of this regularization are that it efficiently suits the structure of the game and that the associated updates are easily computed recursively, starting from the final states.

\subsection{Stabilized OMD algorithm}

The regularizer $\Psi$ sometimes needs to change over time. For example, when $T$ is unknown, a regularizer of the form $\Psi^t=\Psi/\eta^{t}$ is usually considered, with $\eta^t=t^{-1/2}$ the learning rate. \citet{fiegel2023adapting} gives another example of time-varying regularization, adapting the regularization to the game structure that is assumed to be initially unknown.

The previous updates \eqref{eq:omd} do not however allow adaptive regularization in general. In fact, even the simple learning rate decrease $\eta^{t+1}=t^{-1/2}$ can lead to a linear regret dependence with time \citep{orabona2018scale}.

In this part, we shall consider more generally a sequence of Legendre regularizers $(\Psi^t)_{t\in[T]}$ defined on a convex domain $\overline{\Omega}\subset\R^n$, and that the player chooses a sequence of primal iterates $(w^t)_{t\in[T]}$ (respectively the updated realization plans $(\mu^t_{1:})_{t\in[T]}$ of our settings) in a closed convex set $\mathcal{C}$ (respectively the treeplex $Q_{\mathrm{min}}$) included in $\overline{\Omega}$, according to a sequence of dual iterates $(\xi^t)_{t\in[T]}$ in $\R^{n}$ (respectively the estimated losses $(\hell^t)_{t\in[T]}$) observed sequentially.

\citet{fang2020online} proposed, in the non-increasing learning rates case $\Psi^{t+1}=\Psi/\eta^{t+1}$, to use a dual-stabilization to recover the classical OMD bounds. We noticed that their updates can be interpreted as
\begin{equation}
    w^{t+1}=\argmin_{w\in\mathcal{C}} \scal{\xi^t}{w}+\pa{{1}/{\eta^t}}\,\breg_{\Psi}\pa{w,w^t}+\pa{{1}/{\eta^{t+1}}-{1}/{\eta^t}}\breg_{\Psi}\pa{w,w^1}\, .\tag{DS-OMD}
\end{equation}
In the more general adaptive case, we demonstrate that these types of updates can be generalized to
\begin{equation}\label{eq:gsmd}
    w^{t+1}=\argmin_{w\in\mathcal{C}} \scal{\xi^t}{w}+\breg_{\Psi^t}\pa{w,w^t}+\breg_{\Psi^{t+1}-\Psi^t}\pa{w,w^1}\tag{GDS-OMD}
\end{equation}
in which the $\Psi^{t+1}-\Psi^t$ incremental functions are assumed to be convex (but not necessarily Legendre). The following theorem, proven in Appendix~\ref{app:gds} shows that classical OMD guarantees can be recovered with these updates.
\begin{theorem}\label{thm:general_entropy}
Let $(w^t)_{t\in[T]}$ be a sequence of primal iterates generated by the updates \eqref{eq:gsmd}, with convex incremental functions. Then for any $w^\dagger\in \overline{\Omega}$,
\[\sum_{t=1}^T\scal{\xi^t}{w^t-w^\dagger}\leq \vphantom{\sum_{t=1}^T}\breg_{\Psi^{T}}(w^\dagger,w^1)+\sum_{t=1}^T\breg_{\Psi^{t,\star}}\pa{\nabla\Psi^t(w^t)-\xi^t,\nabla\Psi^t(w^t)}\]

where the $\pa{\Psi^{t,\star}}_{t\in[T]}$ are the respective Fenchel conjugates of the $\pa{\Psi^t}_{t\in[T]}$.
\end{theorem}

\begin{remark}
    An interesting example of the use of a regularizer with convex increments (and not only through a decreasing learning rate) is \texttt{AdaGrad} for stochastic gradient descent \citep{duchi2011adaptive}. It uses the adaptive regularization $\Psi^{t+1}=\norm{\cdot}^2_{\pa{G^t}^{1/2}}$, where $G^t$ is a positive semi-definite matrix defined with the gradients $g_k$ by either $G^t=\sum_{k=1}^t g_k^{} g_k^T$  or by the less computationally expensive $G^t=\textrm{Diag}\pa{\sum_{k=1}^t g^{}_k g_k^T}$.
\end{remark}

\paragraph{Adaptive dilatation} In the extensive-form game setting based on the dilated entropy $\Psi^{\mathrm{dil}}_{\alpha}$, this stabilization can be applied to have weights $(\alpha^t(x))_{x\in\cX,t\in[T]}$ that vary with times. The convexity assumption of $\Psi^{\mathrm{dil}}_{\alpha^{t+1}}-\Psi^{\mathrm{dil}}_{\alpha^t}$ then rewrites to having locally non-decreasing weights for each $x\in\cX$. In this particular case, the updates are obtained with the formula
\[
    \label{eq:gsmd_dilated}\mu^{t+1}=\argmin_{\mu\in\maxpi} \scal{\hell^t}{\mu^{}_{1:}}+\breg^\mathrm{dil}_{\alpha^t}(\mu,\mu^t)+\breg^\mathrm{dil}_{\alpha^{t+1}-\alpha^{t}}(\mu,\mu^1)\, .\]
\section{\texorpdfstring{\LocalOMD}{LocalOMD} algorithm}

In this section, we present and analyze the \LocalOMD algorithm.

\subsection{Algorithm}

Let us now consider the fixed sampling framework introduced in Section~\ref{sec:fixed_sampling}. Given a sequence $(\eta^t(x))_{t\in[T]}$ of locally non-increasing learning rates for each $x\in\cX$, we propose to use \OurAlgorithm, based on the updates \label{eq:update} with the adaptive weights $\alpha^t(x)=1/(\mu^s_{1:}(x)\eta^t(x))$ as explained above.

\begin{algorithm}[t] 
\caption{\OurAlgorithm}
\label{alg:omd_fixed}
\begin{algorithmic}[1]
			\STATE \textbf{Input:}\\
			~~~~ Sampling policy $\mu^s\in\maxpi$ and initial policy $\mu^1\in\maxpi$\\
                ~~~~ Bregman divergences $\breg_x$ for each information set $x\in\cX$\\
                ~~~~ $\texttt{UPDATE}\pa{t,x}$ functions that output the non-increasing (adaptive) learning rates $\eta^{t+1}(x)$ after each round $t$ for each information set $x$.

                \STATE \textbf{Output:}\\
                ~~~~ The time-averaged policy $\overline{\mu}$

                \STATE \textbf{Algorithm:}\\

                For $t=1$ to $T$\\
                ~~~~ Observes the outcome of an episode using the fixed strategy $\mu^s$\\
                ~~~~ $q_{H+1}^t\gets 0$\\\vspace{.07cm}
                ~~~~ For $h=H$ to $1$:\\\vspace{.05cm}
                ~~~~~~~~ $\eta^{t+1}(x_h^t)\gets \texttt{UPDATE}\pa{t,x_h^t}$\\
                ~~~~~~~~ $\tell_h^t\gets \left.\indic{a=a_h^t}\pa{{\ell_h^t+q_{h+1}^t}}\middle/{\mu^s(a_h^t|x_h^t)}\right.$\\\vspace{.05cm}
                ~~~~~~~~ $\mu^{t+1}(\cdot | x)\gets\argmin_{\mu\in \Delta_{\cA(x)}} h^t_x(\mu)$ and $q_h^t\gets\min_{\mu\in\Delta_{\cA(x)}}h^t_x(\mu)$\\\vspace{.02cm}
                ~~ where $h^t_x(\mu):=\scal{ \tell_h^t}{\mu}+\frac{1}{\eta^t(x_h^t)}\breg_x\pa{\mu,\mu^t(\cdot | x_h^t)}+\pa{\frac{1}{\eta^{t+1}(x_h^t)}-\frac{1}{\eta^t(x_h^t)}}\breg_x\pa{\mu,\mu^1(\cdot | x_h^t)}$\\\vspace{.1cm}
                ~~~~ For all non-visited $x\in\cX$:\\\vspace{.02cm}
                ~~~~~~~~ $\mu^{t+1}(\cdot | x)\gets \mu^{t}(\cdot | x)$
\end{algorithmic}
\end{algorithm}

\paragraph{Regularized loss} This algorithm can be interpreted as one that locally applies the updates \eqref{eq:gsmd}, but using the loss $\tell_h^t$, a regularized version of the sum of subsequent losses. Even though this algorithm results from a global minimization procedure, the regularized loss has the benefits of only using the probability $\mu^s(a|x)$ of choosing the last action $a\in\cAx$ in the important sampling, instead of the combined probability $\mu_{1:}^s(x,a)$ of the realization plan.

\subsection{Theoretical analysis}

The analysis of \LocalOMD, detailed in Appendix~\ref{app:localomd} is derived from Theorem~\ref{thm:general_entropy} that bounds the estimated regret. The results on the real regret are then obtained with Theorem~\ref{thm:estimation}. We now present two choices of regularization and their associated guarantees.

\paragraph{Optimal rates}

The following theorem uses a constant learning rate that locally depends on the $\kappa(\mu^s|x)$ quantities of Remark~\ref{rmk:kappa}, and on the  $A:=\max_{x\in\cX}\abs{\cAx}$ quantity that upper bounds the local number of available actions on the whole tree.

\begin{theorem}\label{thm:optimal_main}
    Using \LocalOMD with $\mu^1$ as the uniform policy, with the learning rates $\eta^t(x)= {\eta}/{\kappa(\mu^s | x)}$ where $\eta=\sqrt{\log(A){\kappa(\mu^s)/}{(3HT)}}$, and with $\Psi_x$ the Shannon entropy $\Psi_x(\mu)=\sum_{a\in\cAx}\mu(a)\log(\mu(a))$, the regret is bounded with a probability at least $1-\delta$ by 
    \[\regret^T_{\mathrm{min}}\leq \pa{4+2\sqrt{3}}\,H^{3/2}\sqrt{\log(A)\iota\kappa(\mu^s) T}\quad\textrm{where}\quad \iota=\log(2(\AX+1)/\delta)\,.\]
\end{theorem}

When using the balanced policy $\mu^\star$ as the sampling policy, for which $\kappa(\mu^\star)=\AX$, we obtain the rate $\tcO\pa{H^{3/2}\sqrt{\AX T}}$, near-optimal up to the $H$ dependency \citep{bai2022nearoptimal}.

\paragraph{Adaptive rates}

As \OurAlgorithm treats each information set $x\in\cX$ as a separate problem through the regularized losses $\tell_h^t$, one interesting choice is to consider the same adaptive rates that would be used for instance in the $K$-armed bandit problems. The following theorem provides an upper bound in this case.

\begin{theorem}(Informal, exact statement in Appendix~\ref{app:localomd})\\\label{thm:adaptive_main}
    For a large class of regularizers $(\Psi_x)_{x\in\cX}$ and learning rates $(\eta^t(x)_{x\in\cX,t\in[T]})$, the regret has a $\cO(\sqrt{T\log(1/\delta}))$ upper bound (hiding the game-dependent terms) with a probability at least $1-\delta$. Such learning rates include, for all $x\in\cX$ of depth $h$,
    \[\eta^t(x)=\left.\eta\middle/\sqrt{\sum_{k=1}^t \indic{x=x_h^k}}\right.\quad\textrm{or the adaptive version}\quad\eta^t(x)=\left.\eta\middle/\sqrt{\sum_{k=1}^t \indic{x=x_h^k}\pa{\tell_h^k}^2}\right.\, .\]
\end{theorem}

The adaptive learning rates mentioned for this theorem generally enjoy better performances in practice. Furthermore, they require no initial computation and are easily updated.
\section{Experiments}
We implemented \OurAlgorithm, with the parameters of Theorem~\ref{thm:optimal_main} and Theorem~\ref{thm:adaptive_main}, then tested it against the theoretically optimal \texttt{BalancedCFR} \citep{bai2022nearoptimal} and \texttt{BalancedFTRL} \citep{fiegel2023adapting}. The algorithms were compared on three standard benchmark games: Kuhn poker \citep{kuhn1950extensive}, Leduc poker \citep{southey2005bayes} and liars dice, using the OpenSpiel library \citep{openspiel}, with learning rates optimized independently for each algorithm using a grid search. The code is available at \url{https://github.com/anon0493/LocalOMD-experiments}.

The results are given with respect to the total number of episodes used for learning. This technically disadvantages the fixed sampling algorithms, as these require more than one episode at each round $t$ while still performing a single update on the policy of each player.

We observe that the two versions of \OurAlgorithm behave similarly and constantly beat \texttt{BalancedCFR}, mainly because the latter needs to update each depth with independent samples, thus needing $H$ times more episodes overall. The results of \texttt{BalancedFTRL} are more comparable, exhibiting for example better performances on liars dice but worse on Leduc poker.

\begin{figure*}[ht]
\centering
\includegraphics[width=1\textwidth]{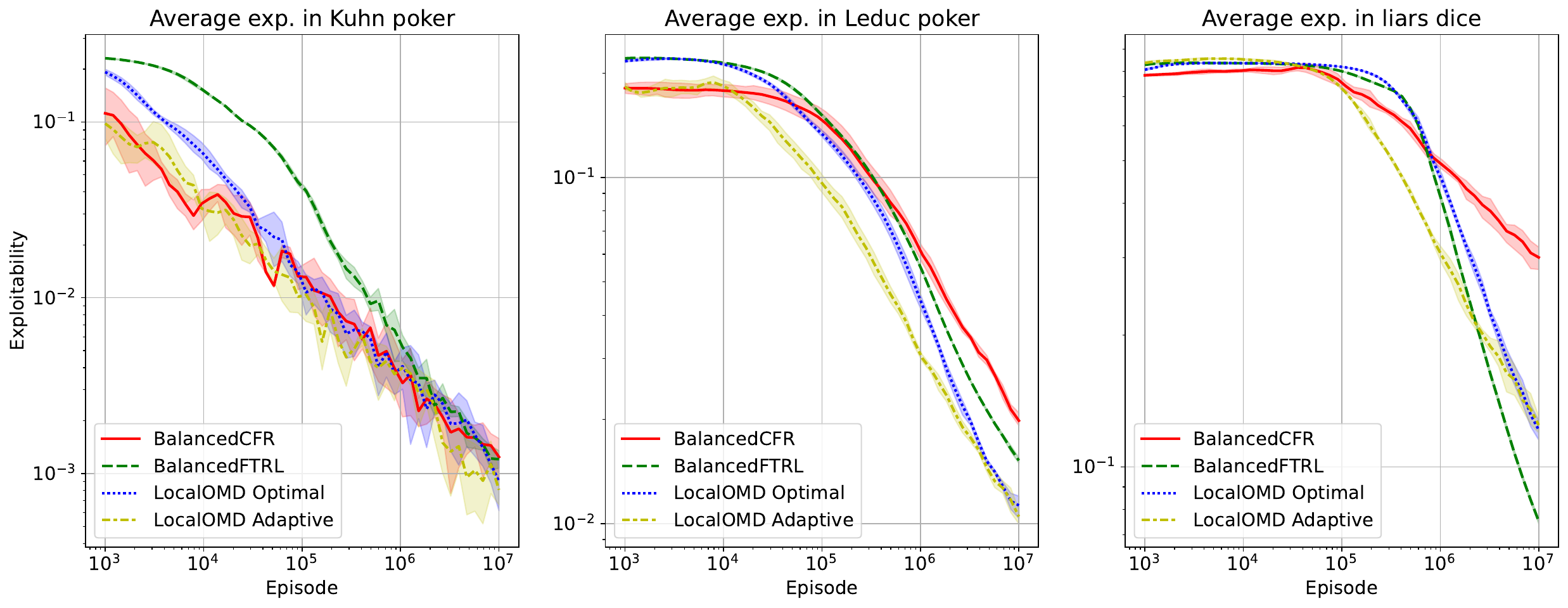}
\caption{Performances of various algorithms with respect to the total number of episodes. The vertical axis denotes the exploitability gap $\max_{(\mu,\nu)\in\maxpi\times\minpi} V^{\overline{\mu},\nu}-V^{\mu,\overline{\nu}}$, with all rewards scaled between $0$ and $1$. The total numbers of actions are $\AX=\BY=12$ for Kuhn poker, $\AX=\BY=1092$ for Leduc poker, and $\AX=\BY=24570$ for Liars dice.}
\label{sec:experiments}
\end{figure*}

\section{Conclusion}

We studied the use of a fixed sampling OMD procedure for the computation of $\epsilon$-optimal strategies. This approach relies, for each player, on an uncoupling between the observation policy and the interaction policy as described in Algorithm~\ref{alg:learning}. This uncoupling is in direct contrast with the more restrictive semi-bandit setting usually considered for self-play, where these two policies must coincide by design. Notice that this is not the standard exploration/exploitation tradeoff, as even in bandit will full monitoring, exploration is required. Seen from an optimization perspective, the two policies indeed influence two different parts of the problem: the \textit{primal iterates} (a simple representation of the interaction policies) and the \textit{dual iterates} (the estimated losses obtained through the observation policies). This distinction between the observation and the interaction was also considered in online convex optimization \citep{bach2016highly}.

While the balanced observation policy recovers the optimal rates in theory, the choice is not as straightforward for large game solving, which requires function approximation. Indeed, the size of each game sub-trees is not as relevant in this case and furthermore, the balanced policy becomes potentially expensive to compute at a given information set. A more practical choice, outside of the current framework, would be to instead use for the observations the current average policy \citep{gibson2012efficient}. This choice could still be adapted to a fixed sampling nonetheless, by restarting the algorithm after a certain number of episodes and using the computed average as the new sampling policy.

We would like to conclude by providing the following interesting research directions.
\paragraph{Problem-dependent optimality} For a given game structure and fixed sampling policy $\mu^s$, is there a policy-dependent lower bound $\cO(\sqrt{\kappa(\mu^s)T})$ on the regret? We wonder if the $\kappa(\mu^s)$ quantity of Remark~\ref{rmk:kappa} denotes some sort of complexity related to the problem.
\paragraph{On-policy algorithms} Is it also possible to remove the importance-sampling of the previous actions in the usual semi-bandit framework that observes with the current policy? The answer is not obvious since the current approach heavily relies on the fact that the sampling policy is fixed.
\paragraph{Last-iterate convergence} Current algorithms need to average the policies updated over time for their guarantees to hold. \citet{daskalakis2018training} shows that, with full information feedback, a convergence of the current policy is possible for normal-form zero-sum games. This result is later extended to extensive-form games by \citet{lee2021last}. Are these types of guarantees obtainable in a trajectory feedback setting? We especially wonder if the fixed sampling approach would help in getting such results.

\section*{Acknowledgements}
P. M\'enard acknowledges the Chaire SeqALO (ANR-20-CHIA-0020-01). Vianney Perchet acknowledges support from the French National Research Agency (ANR) under grant number ANR-19-CE23-0026 as well as from the grant “Investissements d’Avenir” (LabEx Ecodec/ANR-11-LABX-0047).
The authors would like to thank 
Gabriele Farina and Stephen McAleer for their comments and helpful discussions.

\bibliographystyle{plainnat}
\bibliography{ref.bib}

\begin{thebibliography}{59}
\providecommand{\natexlab}[1]{#1}
\providecommand{\url}[1]{\texttt{#1}}
\expandafter\ifx\csname urlstyle\endcsname\relax
  \providecommand{\doi}[1]{doi: #1}\else
  \providecommand{\doi}{doi: \begingroup \urlstyle{rm}\Url}\fi

\bibitem[Auer et~al.(2003)Auer, {Cesa-Bianchi}, Freund, and
  Schapire]{auer2003nonstochastic}
Peter Auer, Nicol{\`o} {Cesa-Bianchi}, Yoav Freund, and Robert~E. Schapire.
\newblock The {{Nonstochastic Multiarmed Bandit Problem}}.
\newblock \emph{SIAM Journal on Computing}, 32\penalty0 (1):\penalty0 48--77,
  January 2003.
\newblock ISSN 0097-5397.
\newblock \doi{10.1137/S0097539701398375}.

\bibitem[Bach and Perchet(2016)]{bach2016highly}
Francis Bach and Vianney Perchet.
\newblock {Highly-Smooth Zero-th Order Online Optimization Vianney Perchet}.
\newblock In \emph{{Conference on Learning Theory (COLT)}}, New York, United
  States, June 2016.
\newblock URL \url{https://hal.science/hal-01321532}.

\bibitem[Bai et~al.(2020)Bai, Jin, and Yu]{bai2020near}
Yu~Bai, Chi Jin, and Tiancheng Yu.
\newblock Near-optimal reinforcement learning with self-play.
\newblock In H.~Larochelle, M.~Ranzato, R.~Hadsell, M.F. Balcan, and H.~Lin,
  editors, \emph{Advances in Neural Information Processing Systems}, volume~33,
  pages 2159--2170. Curran Associates, Inc., 2020.
\newblock URL
  \url{https://proceedings.neurips.cc/paper/2020/file/172ef5a94b4dd0aa120c6878fc29f70c-Paper.pdf}.

\bibitem[Bai et~al.(2022)Bai, Jin, Mei, and Yu]{bai2022nearoptimal}
Yu~Bai, Chi Jin, Song Mei, and Tiancheng Yu.
\newblock Near-optimal learning of extensive-form games with imperfect
  information.
\newblock In \emph{{{International Conference}} on {{Machine Learning}}}, 2022.

\bibitem[Burch et~al.(2019)Burch, Moravčík, and Schmid]{burch2019revisiting}
Neil Burch, Matej Moravčík, and Martin Schmid.
\newblock {Revisiting CFR+ and Alternating Updates}.
\newblock \emph{Journal of Artificial Intelligence Research}, 64:\penalty0
  429--443, 2019.

\bibitem[{Cesa-Bianchi} and Lugosi(2006)]{cesa-bianchi2006prediction}
Nicolo {Cesa-Bianchi} and Gabor Lugosi.
\newblock \emph{Prediction, {{Learning}}, and {{Games}}}.
\newblock {Cambridge University Press}, {Cambridge}, 2006.
\newblock ISBN 978-0-521-84108-5.
\newblock \doi{10.1017/CBO9780511546921}.

\bibitem[Daskalakis et~al.(2018)Daskalakis, Ilyas, Syrgkanis, and
  Zeng]{daskalakis2018training}
Constantinos Daskalakis, Andrew Ilyas, Vasilis Syrgkanis, and Haoyang Zeng.
\newblock Training {GAN}s with optimism.
\newblock In \emph{International Conference on Learning Representations}, 2018.
\newblock URL \url{https://openreview.net/forum?id=SJJySbbAZ}.

\bibitem[Daskalakis et~al.(2020)Daskalakis, Foster, and
  Golowich]{daskalakis2020independent}
Constantinos Daskalakis, Dylan~J Foster, and Noah Golowich.
\newblock Independent policy gradient methods for competitive reinforcement
  learning.
\newblock In H.~Larochelle, M.~Ranzato, R.~Hadsell, M.F. Balcan, and H.~Lin,
  editors, \emph{Advances in Neural Information Processing Systems}, volume~33,
  pages 5527--5540. Curran Associates, Inc., 2020.
\newblock URL
  \url{https://proceedings.neurips.cc/paper/2020/file/3b2acfe2e38102074656ed938abf4ac3-Paper.pdf}.

\bibitem[Duchi et~al.(2011)Duchi, Hazan, and Singer]{duchi2011adaptive}
John Duchi, Elad Hazan, and Yoram Singer.
\newblock Adaptive subgradient methods for online learning and stochastic
  optimization.
\newblock \emph{Journal of Machine Learning Research}, 12\penalty0
  (61):\penalty0 2121--2159, 2011.
\newblock URL \url{http://jmlr.org/papers/v12/duchi11a.html}.

\bibitem[Fang et~al.(2020)Fang, Harvey, Portella, and
  Friedlander]{fang2020online}
Huang Fang, Nick Harvey, Victor Portella, and Michael Friedlander.
\newblock Online mirror descent and dual averaging: Keeping pace in the dynamic
  case.
\newblock In \emph{Proceedings of the 37th {{International Conference}} on
  {{Machine Learning}}}, pages 3008--3017. {PMLR}, November 2020.

\bibitem[Farina et~al.(2019)Farina, Kroer, and Sandholm]{farina2019regret}
Gabriele Farina, Christian Kroer, and Tuomas Sandholm.
\newblock Regret circuits: Composability of regret minimizers.
\newblock In Kamalika Chaudhuri and Ruslan Salakhutdinov, editors,
  \emph{Proceedings of the 36th International Conference on Machine Learning,
  {ICML} 2019, 9-15 June 2019, Long Beach, California, {USA}}, volume~97 of
  \emph{Proceedings of Machine Learning Research}, pages 1863--1872. {PMLR},
  2019.
\newblock URL \url{http://proceedings.mlr.press/v97/farina19b.html}.

\bibitem[Farina et~al.(2020)Farina, Kroer, and Sandholm]{farina2020stochastic}
Gabriele Farina, Christian Kroer, and Tuomas Sandholm.
\newblock {Stochastic Regret Minimization in Extensive-Form Games}.
\newblock In \emph{International Conference on Machine Learning}, 2020.

\bibitem[Farina et~al.(2021{\natexlab{a}})Farina, Kroer, and
  Sandholm]{farina2020faster}
Gabriele Farina, Christian Kroer, and Tuomas Sandholm.
\newblock {Faster Game Solving via Predictive Blackwell Approachability:
  Connecting Regret Matching and Mirror Descent}.
\newblock In \emph{AAAI Conference on Artificial Intelligence},
  2021{\natexlab{a}}.
\newblock URL \url{https://arxiv.org/abs/2007.14358}.

\bibitem[Farina et~al.(2021{\natexlab{b}})Farina, Kroer, and
  Sandholm]{farina2021bandit}
Gabriele Farina, Christian Kroer, and Tuomas Sandholm.
\newblock {Bandit Linear Optimization for Sequential Decision Making and
  Extensive-Form Games}.
\newblock In \emph{AAAI Conference on Artificial Intelligence},
  2021{\natexlab{b}}.

\bibitem[Fiegel et~al.(2023)Fiegel, Ménard, Kozuno, Munos, Perchet, and
  Valko]{fiegel2023adapting}
Côme Fiegel, Pierre Ménard, Tadashi Kozuno, Rémi Munos, Vianney Perchet, and
  Michal Valko.
\newblock Adapting to game trees in zero-sum imperfect information games, 2023.

\bibitem[Gibson et~al.(2012)Gibson, Burch, Lanctot, and
  Szafron]{gibson2012efficient}
Richard Gibson, Neil Burch, Marc Lanctot, and Duane Szafron.
\newblock Efficient monte carlo counterfactual regret minimization in games
  with many player actions.
\newblock volume~3, 12 2012.

\bibitem[Gordon(2007)]{gordon2007no}
Geoffrey~J Gordon.
\newblock {No-regret Algorithms for Online Convex Programs}.
\newblock In \emph{Advances in Neural Information Processing Systems}, 2007.

\bibitem[Hart and Mas-Colell(2000)]{hart2000simple}
Sergiu Hart and Andreu Mas-Colell.
\newblock {A Simple Adaptive Procedure Leading to Correlated Equilibrium}.
\newblock \emph{Econometrica}, 68\penalty0 (5):\penalty0 1127--1150, 2000.

\bibitem[Hoda et~al.(2010)Hoda, Gilpin, Pe{\~{n}}a, and
  Sandholm]{hoda2010smoothing}
Samid Hoda, Andrew Gilpin, Javier Pe{\~{n}}a, and Tuomas Sandholm.
\newblock {Smoothing Techniques for Computing Nash Equilibria of Sequential
  Games}.
\newblock \emph{Mathematics of Operations Research}, 2010.
\newblock URL \url{https://kilthub.cmu.edu/ndownloader/files/12101699}.

\bibitem[Johanson et~al.(2012)Johanson, Bard, Lanctot, Gibson, and
  Bowling]{johanson2012efficient}
Michael Johanson, Nolan Bard, Marc Lanctot, Richard Gibson, and Michael
  Bowling.
\newblock Efficient nash equilibrium approximation through monte carlo
  counterfactual regret minimization.
\newblock In \emph{Proceedings of the 11th International Conference on
  Autonomous Agents and Multiagent Systems - Volume 2}, AAMAS '12, page
  837–846, Richland, SC, 2012. International Foundation for Autonomous Agents
  and Multiagent Systems.
\newblock ISBN 0981738125.

\bibitem[Koller et~al.(1996)Koller, Megiddo, and
  Von~Stengel]{koller1996efficient}
Daphne Koller, Nimrod Megiddo, and Bernhard Von~Stengel.
\newblock {Efficient Computation of Equilibria for Extensive Two-Person Games}.
\newblock \emph{Games and Economic Behavior}, 14\penalty0 (2):\penalty0
  247--259, 1996.

\bibitem[Kozuno et~al.(2021)Kozuno, M{\'e}nard, Munos, and
  Valko]{kozuno2021learning}
Tadashi Kozuno, Pierre M{\'e}nard, Remi Munos, and Michal Valko.
\newblock Learning in two-player zero-sum partially observable {{Markov}} games
  with perfect recall.
\newblock In \emph{Neural Information Processing Systems}, 2021.

\bibitem[Kroer et~al.(2015)Kroer, Waugh, {Kilin{\c c}-Karzan}, and
  Sandholm]{kroer2015faster}
Christian Kroer, Kevin Waugh, Fatma {Kilin{\c c}-Karzan}, and Tuomas Sandholm.
\newblock Faster first-order methods for extensive-form game solving.
\newblock In \emph{Economics and Computation}, 2015.
\newblock ISBN 978-1-4503-3410-5.
\newblock \doi{10.1145/2764468.2764476}.

\bibitem[Kroer et~al.(2018)Kroer, Farina, and Sandholm]{kroer2018solving}
Christian Kroer, Gabriele Farina, and Tuomas Sandholm.
\newblock Solving large sequential games with the excessive gap technique.
\newblock In \emph{Neural Information Processing Systems}, 2018.

\bibitem[Kroer et~al.(2020)Kroer, Waugh, K{\i}l{\i}n{\c{c}}-Karzan, and
  Sandholm]{kroer2020faster}
Christian Kroer, Kevin Waugh, Fatma K{\i}l{\i}n{\c{c}}-Karzan, and Tuomas
  Sandholm.
\newblock Faster algorithms for extensive-form game solving via improved
  smoothing functions.
\newblock \emph{Mathematical Programming}, 179\penalty0 (1):\penalty0 385--417,
  2020.

\bibitem[Kuhn(1950)]{kuhn1950extensive}
Harold~W Kuhn.
\newblock Extensive games.
\newblock \emph{Proceedings of the National Academy of Sciences}, 36\penalty0
  (10):\penalty0 570--576, 1950.

\bibitem[Kuhn(1953)]{kuhn1953extensive}
Harold~W Kuhn.
\newblock {Extensive Games and the Problem of Information}.
\newblock \emph{Annals of Mathematics Studies}, 28:\penalty0 193--216, 1953.

\bibitem[Lanctot et~al.(2009)Lanctot, Waugh, Zinkevich, and
  Bowling]{lanctot2009monte}
Marc Lanctot, Kevin Waugh, Martin Zinkevich, and Michael Bowling.
\newblock {M}onte-{C}arlo sampling for regret minimization in extensive games.
\newblock In \emph{Neural Information Processing Systems}, 2009.

\bibitem[Lanctot et~al.(2019)Lanctot, Lockhart, Lespiau, Zambaldi, Upadhyay,
  Pérolat, Srinivasan, Timbers, Tuyls, Omidshafiei, Hennes, Morrill, Muller,
  Ewalds, Faulkner, Kramár, De~Vylder, Saeta, Bradbury, Ding, Borgeaud, Lai,
  Schrittwieser, Anthony, Hughes, Danihelka, and Ryan-Davis]{openspiel}
Marc Lanctot, Edward Lockhart, Jean-Baptiste Lespiau, Vinicius Zambaldi,
  Satyaki Upadhyay, Julien Pérolat, Sriram Srinivasan, Finbarr Timbers, Karl
  Tuyls, Shayegan Omidshafiei, Daniel Hennes, Dustin Morrill, Paul Muller, Timo
  Ewalds, Ryan Faulkner, János Kramár, Bart De~Vylder, Brennan Saeta, James
  Bradbury, David Ding, Sebastian Borgeaud, Matthew Lai, Julian Schrittwieser,
  Thomas Anthony, Edward Hughes, Ivo Danihelka, and Jonah Ryan-Davis.
\newblock Openspiel: A framework for reinforcement learning in games, 2019.
\newblock URL \url{https://arxiv.org/abs/1908.09453}.

\bibitem[Laraki et~al.(2019)Laraki, Renault, and Sorin]{larakibookgame}
Rida Laraki, J{\'e}r{\^o}me Renault, and Sylvain Sorin.
\newblock \emph{{Mathematical Foundations of Game Theory}}.
\newblock {Springer}, October 2019.
\newblock \doi{10.1007/978-3-030-26646-2}.
\newblock URL \url{https://hal.science/hal-03070434}.

\bibitem[Lattimore and Szepesv{\'a}ri(2020)]{lattimore2020bandit_book}
Tor Lattimore and Csaba Szepesv{\'a}ri.
\newblock \emph{Bandit Algorithms}.
\newblock Cambridge University Press, 2020.
\newblock \doi{10.1017/9781108571401}.

\bibitem[Lee et~al.(2021)Lee, Kroer, and Luo]{lee2021last}
Chung-Wei Lee, Christian Kroer, and Haipeng Luo.
\newblock Last-iterate convergence in extensive-form games.
\newblock In \emph{Neural Information Processing Systems}, 2021.
\newblock URL
  \url{https://proceedings.neurips.cc/paper/2021/file/77bb14f6132ea06dea456584b7d5581e-Paper.pdf}.

\bibitem[Liu et~al.(2021)Liu, Yu, Bai, and Jin]{liu2021sharp}
Qinghua Liu, Tiancheng Yu, Yu~Bai, and Chi Jin.
\newblock A sharp analysis of model-based reinforcement learning with
  self-play.
\newblock In Marina Meila and Tong Zhang, editors, \emph{Proceedings of the
  38th International Conference on Machine Learning}, volume 139 of
  \emph{Proceedings of Machine Learning Research}, pages 7001--7010. PMLR,
  18--24 Jul 2021.
\newblock URL \url{https://proceedings.mlr.press/v139/liu21z.html}.

\bibitem[McAleer et~al.(2022)McAleer, Farina, Lanctot, and
  Sandholm]{mcaleer2022escher}
Stephen McAleer, Gabriele Farina, Marc Lanctot, and Tuomas Sandholm.
\newblock {ESCHER:} {E}schewing importance sampling in games by computing a
  history value function to estimate regret.
\newblock \emph{CoRR}, abs/2206.04122, 2022.
\newblock \doi{10.48550/arXiv.2206.04122}.
\newblock URL \url{https://doi.org/10.48550/arXiv.2206.04122}.

\bibitem[McMahan(2017)]{McMahansurvey}
H.B. McMahan.
\newblock A survey of algorithms and analysis for adaptive online learning.
\newblock \emph{Journal of Machine Learning Research}, 18:\penalty0 1--50, 08
  2017.

\bibitem[Munos et~al.(2020)Munos, P\'erolat, Lespiau, Rowland, De~Vylder,
  Lanctot, Timbers, Hennes, Omidshafiei, Gruslys, Azar, Lockhart, and
  Tuyls]{munos2020fast}
R\'emi Munos, Julien P\'erolat, Jean-Baptiste Lespiau, Mark Rowland, Bart
  De~Vylder, Marc Lanctot, Finbarr Timbers, Daniel Hennes, Shayegan
  Omidshafiei, Audrunas Gruslys, Mohammad~Gheshlaghi Azar, Edward Lockhart, and
  Karl Tuyls.
\newblock Fast computation of nash equilibria in imperfect information games.
\newblock In \emph{International Conference on Machine Learning}, 2020.

\bibitem[Nemirovski(2004)]{nemirovski2004prox}
Arkadi Nemirovski.
\newblock {Prox-Method with Rate of Convergence $O(1/t)$ for Variational
  Inequalities with Lipschitz Continuous Monotone Operators and Smooth
  Convex-Concave Saddle Point Problems}.
\newblock \emph{SIAM Journal on Optimization}, 15\penalty0 (1):\penalty0
  229--251, 2004.

\bibitem[Nesterov(2005)]{nesterov2005smooth}
Yu~Nesterov.
\newblock {Smooth minimization of non-smooth functions}.
\newblock \emph{Mathematical programming}, 103\penalty0 (1):\penalty0 127--152,
  2005.

\bibitem[Orabona and P{\'{a}}l(2018)]{orabona2018scale}
Francesco Orabona and D{\'{a}}vid P{\'{a}}l.
\newblock Scale-free online learning.
\newblock \emph{Theor. Comput. Sci.}, 716:\penalty0 50--69, 2018.
\newblock \doi{10.1016/j.tcs.2017.11.021}.
\newblock URL \url{https://doi.org/10.1016/j.tcs.2017.11.021}.

\bibitem[Osborne and Rubinstein(1994)]{osborne1994course}
Martin~J. Osborne and Ariel Rubinstein.
\newblock \emph{{A Course in Game Theory}}.
\newblock The MIT Press, 1994.
\newblock ISBN 0-262-65040-1.

\bibitem[Ponsen et~al.(2011)Ponsen, De~Jong, and Lanctot]{ponsen2011computing}
Marc Ponsen, Steven De~Jong, and Marc Lanctot.
\newblock {Computing Approximate Nash Equilibria and Robust Best-Responses
  Using Sampling}.
\newblock \emph{Journal of Artificial Intelligence Research}, 42:\penalty0
  575--605, 2011.

\bibitem[Romanovsky(1962)]{Rom62}
J.~V. Romanovsky.
\newblock Reduction of a game with complete memory to a matricial game.
\newblock \emph{Dokl. Akad. Nauk SSSR}, 144:\penalty0 62--64, 1962.

\bibitem[Schmid et~al.(2018)Schmid, Burch, Lanctot, Moravčík, Kadlec, and
  Bowling]{shcmid2018variance}
Martin Schmid, Neil Burch, Marc Lanctot, Matej Moravčík, Rudolf Kadlec, and
  Michael Bowling.
\newblock Variance reduction in monte carlo counterfactual regret minimization
  {(VR-MCCFR)} for extensive form games using baselines.
\newblock \emph{CoRR}, abs/1809.03057, 2018.
\newblock URL \url{http://arxiv.org/abs/1809.03057}.

\bibitem[Sidford et~al.(2020)Sidford, Wang, Yang, and Ye]{sidford2020soving}
Aaron Sidford, Mengdi Wang, Lin Yang, and Yinyu Ye.
\newblock Solving discounted stochastic two-player games with near-optimal time
  and sample complexity.
\newblock In Silvia Chiappa and Roberto Calandra, editors, \emph{The 23rd
  International Conference on Artificial Intelligence and Statistics, {AISTATS}
  2020, 26-28 August 2020, Online [Palermo, Sicily, Italy]}, volume 108 of
  \emph{Proceedings of Machine Learning Research}, pages 2992--3002. {PMLR},
  2020.
\newblock URL \url{http://proceedings.mlr.press/v108/sidford20a.html}.

\bibitem[Southey et~al.(2005)Southey, Bowling, Larson, Piccione, Burch,
  Billings, and Rayner]{southey2005bayes}
Finnegan Southey, Michael Bowling, Bryce Larson, Carmelo Piccione, Neil Burch,
  Darse Billings, and Chris Rayner.
\newblock Bayes' bluff: {O}pponent modelling in poker.
\newblock In \emph{Proceedings of the Twenty-First Conference on Uncertaintyin
  Artificial Intelligence (UAI)}, pages 550--558, 2005.

\bibitem[Steinberger et~al.(2020)Steinberger, Lerer, and
  Brown]{steinberger2020dream}
Eric Steinberger, Adam Lerer, and Noam Brown.
\newblock {DREAM:} deep regret minimization with advantage baselines and
  model-free learning.
\newblock \emph{CoRR}, abs/2006.10410, 2020.
\newblock URL \url{https://arxiv.org/abs/2006.10410}.

\bibitem[Strens(2000)]{strens2000bayesian}
Malcolm Strens.
\newblock {A Bayesian Framework for Reinforcement Learning}.
\newblock In \emph{International Conference on Machine Learning}, 2000.

\bibitem[Tammelin(2014)]{tammelin2014solving}
Oskari Tammelin.
\newblock Solving large imperfect information games using {CFR+}.
\newblock \emph{arXiv preprint arXiv:1407.5042}, 2014.

\bibitem[von Neumann(1928)]{neumann1928zur}
J.~von Neumann.
\newblock {Zur Theorie der Gesellschaftsspiele}.
\newblock \emph{Mathematische Annalen}, 100:\penalty0 295--320, 1928.
\newblock ISSN 0025-5831; 1432-1807/e.

\bibitem[{von Stengel}(1996)]{VONSTENGEL1996220}
Bernhard {von Stengel}.
\newblock Efficient computation of behavior strategies.
\newblock \emph{Games and Economic Behavior}, 14\penalty0 (2):\penalty0
  220--246, 1996.

\bibitem[Waugh and Bagnell(2014)]{waugh2014unified}
Kevin Waugh and J.~Andrew Bagnell.
\newblock A unified view of large-scale zero-sum equilibrium computation.
\newblock \emph{CoRR}, abs/1411.5007, 2014.
\newblock URL \url{http://arxiv.org/abs/1411.5007}.

\bibitem[Wei et~al.(2017)Wei, Hong, and Lu]{wei2017online}
Chen-Yu Wei, Yi-Te Hong, and Chi-Jen Lu.
\newblock Online reinforcement learning in stochastic games.
\newblock In I.~Guyon, U.~Von Luxburg, S.~Bengio, H.~Wallach, R.~Fergus,
  S.~Vishwanathan, and R.~Garnett, editors, \emph{Advances in Neural
  Information Processing Systems}, volume~30. Curran Associates, Inc., 2017.
\newblock URL
  \url{https://proceedings.neurips.cc/paper/2017/file/36e729ec173b94133d8fa552e4029f8b-Paper.pdf}.

\bibitem[Wei et~al.(2021)Wei, Lee, Zhang, and Luo]{wei2021last}
Chen{-}Yu Wei, Chung{-}Wei Lee, Mengxiao Zhang, and Haipeng Luo.
\newblock Last-iterate convergence of decentralized optimistic gradient
  descent/ascent in infinite-horizon competitive markov games.
\newblock In Mikhail Belkin and Samory Kpotufe, editors, \emph{Conference on
  Learning Theory, {COLT} 2021, 15-19 August 2021, Boulder, Colorado, {USA}},
  volume 134 of \emph{Proceedings of Machine Learning Research}, pages
  4259--4299. {PMLR}, 2021.
\newblock URL \url{http://proceedings.mlr.press/v134/wei21a.html}.

\bibitem[Xie et~al.(2020)Xie, Chen, Wang, and Yang]{xie2020learning}
Qiaomin Xie, Yudong Chen, Zhaoran Wang, and Zhuoran Yang.
\newblock Learning zero-sum simultaneous-move markov games using function
  approximation and correlated equilibrium.
\newblock In Jacob Abernethy and Shivani Agarwal, editors, \emph{Proceedings of
  Thirty Third Conference on Learning Theory}, volume 125 of \emph{Proceedings
  of Machine Learning Research}, pages 3674--3682. PMLR, 09--12 Jul 2020.
\newblock URL \url{https://proceedings.mlr.press/v125/xie20a.html}.

\bibitem[Zhang and Sandholm(2021)]{zhang2020finding}
Brian~Hu Zhang and Tuomas Sandholm.
\newblock {Finding and Certifying (Near-) Optimal Strategies in Black-Box
  Extensive-Form Games}.
\newblock In \emph{AAAI Conference on Artificial Intelligence}, 2021.

\bibitem[Zhang et~al.(2021)Zhang, Bengio, Hardt, Recht, and
  Vinyals]{zhang2021understanding}
Chiyuan Zhang, Samy Bengio, Moritz Hardt, Benjamin Recht, and Oriol Vinyals.
\newblock Understanding deep learning (still) requires rethinking
  generalization.
\newblock \emph{Commun. ACM}, 64\penalty0 (3):\penalty0 107–115, feb 2021.
\newblock ISSN 0001-0782.
\newblock \doi{10.1145/3446776}.
\newblock URL \url{https://doi.org/10.1145/3446776}.

\bibitem[Zhang et~al.(2020)Zhang, Kakade, Basar, and Yang]{zhang2020model}
Kaiqing Zhang, Sham Kakade, Tamer Basar, and Lin Yang.
\newblock Model-based multi-agent rl in zero-sum markov games with near-optimal
  sample complexity.
\newblock In H.~Larochelle, M.~Ranzato, R.~Hadsell, M.F. Balcan, and H.~Lin,
  editors, \emph{Advances in Neural Information Processing Systems}, volume~33,
  pages 1166--1178. Curran Associates, Inc., 2020.
\newblock URL
  \url{https://proceedings.neurips.cc/paper/2020/file/0cc6ee01c82fc49c28706e0918f57e2d-Paper.pdf}.

\bibitem[Zhou et~al.(2020)Zhou, Li, and Zhu]{zhou2020posterior}
Yichi Zhou, J.~Li, and J.~Zhu.
\newblock {Posterior sampling for multi-agent reinforcement learning: solving
  extensive games with imperfect information}.
\newblock In \emph{International Conference on Learning Representations}, 2020.
\newblock URL \url{https://openreview.net/pdf?id=Syg-ET4FPS}.

\bibitem[Zinkevich et~al.(2007)Zinkevich, Johanson, Bowling, and
  Piccione]{zinkevich2007regret}
Martin Zinkevich, Michael Johanson, Michael Bowling, and Carmelo Piccione.
\newblock Regret minimization in games with incomplete information.
\newblock \emph{Neural Information Processing Systems}, 2007.

\end{thebibliography}

\newpage
\appendix

\section{Related works}

In this section, we review previous works on learning an $\epsilon$-optimal strategy in IIGs.

\paragraph{Full feedback} When the game is known, that is the information set structure space, transitions probability, and reward function are provided, a first line of work recasts the setting through the sequence-form representation of a game as a linear program which can be
solved efficiently \citep{Rom62,VONSTENGEL1996220,koller1996efficient}. A second line of work relies on first-order optimization methods for saddle point computation \citep{hoda2010smoothing, kroer2015faster,kroer2018solving,kroer2020faster,munos2020fast,lee2021last}.
In particular \citet{hoda2010smoothing,kroer2018solving} relies on the Nesterov smoothing technique \citet{nesterov2005smooth} whereas \citet{kroer2015faster,kroer2020faster} use the \MirrorProx algorithm \citep{nemirovski2004prox}. These methods have a rate of convergence of order $\tcO(\poly(H,\AX,\BY)/\epsilon)$.

A third approach, counterfactual regret minimization \citep{zinkevich2007regret}, leverages local regret minimization, i.e. minimizing a type of regret at each information set. Popular algorithms are based on the regret-matching algorithm \citep{hart2000simple,gordon2007no} such as \CFR algorithm \citep{zinkevich2007regret} or based on a close variant of regret-matching, e.g. \CFRp \citep{tammelin2014solving, burch2019revisiting,farina2020faster}. Note that other local regret minimizers could be used, see for example \citet{ waugh2014unified,farina2019regret}. These algorithms enjoy a guarantee of convergence of order $\tcO(\poly(H,\AX,\BY)/\epsilon^2)$. 

Nevertheless, all the methods described above need to explore \emph{the whole information set tree} (or the whole state space) in order to compute one update. The cost of one traversal is of order $\cO(X+Y)$ if the transitions and the actions of the other player are sampled; see for example the
external-sampling \MCCFR algorithm \citep{lanctot2009monte}.

\paragraph{Trajectory feedback} A way to tackle the aforementioned issues is to consider the agnostic setting where the \emph{agent
has no prior knowledge of the game and only observes trajectories of the game}. Precisely, the rewards and the transition probabilities are unknown.

\paragraph{Model-based} A first method to deal with this limited feedback is to build a \emph{model} of the game and then run any full feedback algorithm in this model. For example, \citet{zhou2020posterior} use \textit{posterior sampling} (PS, \citealp{strens2000bayesian}) to learn a model and then use the \CFR algorithm in games sampled from the posterior. They obtain a convergence rate of order $\tcO(\poly(H,S,A,B)/\epsilon^2)$ but only when the games are actually sampled according to the known prior. Instead, \citet{zhang2020finding} relies on the principle of optimism in the presence of uncertainty to incrementally build a model of the game. Then, the \CFR algorithm is fed with \emph{optimistic estimates} of the local regrets. They prove a high-probability sample complexity of order $\tcO(\poly(H,S,A,B)/\epsilon^2)$.

\paragraph{Model-free} Another line of work \citep{lanctot2009monte,johanson2012efficient,shcmid2018variance,farina2020stochastic} directly estimates the local regret via importance sampling that is then fed to the \CFR algorithm. In particular, the outcome-sampling \MCCFR \citep{lanctot2009monte, farina2020stochastic} builds an importance sampling estimate of the counterfactual regret by playing according to a well-chosen \emph{balanced policy}. Intuitively, this policy should ensure to \emph{explore all the information sets}. Note that, depending on the structure of the information set space, playing uniformly over the actions at each information set is not necessarily a good choice. Instead, \citet{farina2020stochastic} propose as a balanced policy to play action with probability proportional to the number of leaves in the sub-tree of possible next information sets. In particular, the outcome-sampling \MCCFR algorithm requires the knowledge of the information set space structure to build its balanced policy. Nonetheless, in order to obtain $\epsilon$-optimal strategies with high probability, \MCCFR needs at most
 $\tcO(H^3(\AX+\BY )/\epsilon^2)$ realizations of the game \citep{farina2020stochastic,bai2022nearoptimal}. 

Later, \citet{kozuno2021learning} proposed to combine \emph{Online Mirror Descent (\OMD)} with \emph{dilated Shannon entropy as regularizer} and importance sampling estimate of the losses of a player, see also \citet{farina2021bandit}. They prove a sample complexity, for the  proposed algorithm, \IXOMD, of order $\tcO(H^2(X\AX+Y\BY  )/\epsilon^2)$. Interestingly, they do not need to know in advance the structure of the information set space to obtain this bound. However, the sample complexity of \IXOMD does not match the lower bound for this setting which is of order $\cO((\AX+\BY)/\epsilon^2)$. Recently, \citet{bai2022nearoptimal} proposed the \BalancedOMD algorithm that enjoys also relies on \OMD but with a dilated entropy weighted by the realization plans of balanced policies as regularizers. For this algorithm, they prove a sample complexity of order $\tcO(H^3(\AX+\BY)/\epsilon^2)$.

\paragraph{Perfect information Markov game} Another line of work considers Markov game \citet{kuhn1953extensive} with \emph{perfect} information and limited feedback. However, it does not assume perfect recall. \citet{sidford2020soving,zhang2020model,daskalakis2020independent,wei2021last} consider the case where a \emph{generative model} is available whereas \citet{wei2017online,bai2020near,xie2020learning,liu2021sharp} deal with the \emph{trajectory feedback} case. Although this setting is related to ours there is no direct comparison between the two.

\section{Regret estimation} \label{app:approximation}

In this section, we aim to establish Theorem~\ref{thm:estimation} of the main paper. We start by stating a Bernstein-type inequality that we will use multiple times. It can be found e.g. in Exercise~5.15 by \citet{lattimore2020bandit_book}. We provide a short proof below as we did not find any for this precise statement.

\begin{lemma}\label{lemma_concentration}
    Let $(U^t)_{t\in [T]}$ be a sequence of random variables with respect to a filtration $\salgebra$, and $\gamma>0$ be a fixed constant such that for all $t$, $\gamma U^t\leq 1$. Then with a probability of at least $1-\delta'$:

    \[\sum_{t=1}^T \pa{U^t-\E\bra{U^t \middle| \salgebra^{t-1}}}\leq \gamma \sum_{t=1}^T \E\bra{(U^t)^2 \middle| \salgebra^{t-1}} + \frac{1}{\gamma}\log(\frac{1}{\delta'})\]
\end{lemma}

\begin{proof}
    For any $t\in[T]$, using the inequalities $\exp(x)\leq 1+x+x^2$ for all $x\leq 1$ and $1+x\leq \exp(x)$ for all $x\in\R$, we have
    \begin{align*}
        \E\bra{\exp\pa{\gamma U^t}\middle| \salgebra^{t-1}}&\leq \E\bra{1+\gamma U^t+\gamma^2 (U^t)^2\middle|\salgebra^{t-1}}\\
        &=1+\gamma \E\bra{U^t\middle| \salgebra^{t-1}}+ \gamma^2\E\bra{(U^t)^2\middle|\salgebra^{t-1}}\\
        &\leq \exp\pa{\gamma\E\bra{U^t\middle| \salgebra^{t-1}}+ \gamma^2\E\bra{(U^t)^2\middle|\salgebra^{t-1}}}\, .
    \end{align*}

    This implies that the random process $(S_t)_{t\in[T]}$ defined by
    \[S_t:=\exp\pa{\sum_{k=1}^t\gamma\pa{U^k-\E\bra{U^k\middle| \salgebra^{k-1}}}-\sum_{k=1}^t\gamma^2\E\bra{(U^k)^2\middle|\salgebra^{k-1}}}\]
    is a super-martingale, with $S_0=1$. Using the Markov inequality, we then get
    \[\P\pa{\frac{1}{\gamma}\log(S_T) > \frac{1}{\gamma}\log\pa{\frac{1}{\delta'}}}=\P\pa{S_T> \frac{1}{\delta'}}\leq \delta'\,\E(S_T)\leq \delta'\]

    which immediately yields the stated inequality with probability at least $1-\delta'$.
\end{proof}

This lemma is then used for Theorem~\ref{thm:estimation}. The filtration $(\salgebra^t)_{t\in[T]}$ will be used, such that $\salgebra^t$ is the sigma-algebra of all variables of the self-play algorithm up to the execution of episode $t+1$.

\begin{theorem}
    Assume that the estimated losses are obtained with a fixed positive sampling policy $\mu^s$ as above. Then, for any sequence $(\mu^t)_{t\in [T]}$ of $\maxpi$ and any $\delta\in(0,1)$, the following bound holds with a probability at least $1-\delta$
    \[\regret^T_\textrm{min}\leq \max\left\{\hat{\regret}^T_\textrm{min},0\right\}+ 4\sqrt{ \iota H\kappa(\mu^s) T}\]
    where
    \[\iota:=\log\pa{\frac{\AX+1}{\delta}}\quad\textrm{and}\quad\kappa(\mu^s):=\max_{\mu\in\maxpi}\sum_{x\in\cX}\sum_{a\in\cA_x}\frac{\mu_{1:}(x,a)}{\mu_{1:}^s(x,a)}\, .\]
\end{theorem}

\begin{proof}
We want to show that, with probability at least $1-\delta$, that
    \[\sum_{t=1}^t\scal{\ell^t-\hell^t}{\mu_{1:}^t-\mu_{1:}}\leq 4\sqrt{ \iota H\kappa(\mu^s) T}\]
holds for all $\mu\in\maxpi$. Then the property follows after re-organizing the inequality and maximizing over $\mu$. In order to do so, we divide this term into two parts:
\[\sum_{t=1}^T\scal{\ell^t-\hell^t}{\mu_{1:}^t-\mu_{1:}}=\underbrace{\sum_{t=1}^T \scal{\hell^t-\ell^t}{\mu_{1:}}}_{\textrm{EST I}}+\underbrace{\sum_{t=1}^T \scal{\ell^t-\hell^t}{\mu_{1:}^t}}_{\textrm{EST II}}\,.\]

We will furthermore assume that $HT\geq\iota\kappa(\mu^s)$, as otherwise, $4\sqrt{\iota H\kappa(\mu^s)T}\leq 4HT$ and the property immediately follows from $\regret^T_\textrm{min}\leq HT$.

\textit{Upper bound of EST I} For all $x\in\cX$ of depth $h$ and $a\in\cA(x)$, we apply Lemma~\ref{lemma_concentration} to the random process 
    \[U^t_{x,a}=\ell_h^t\indic{x=x_h^t,a=a_h^t}\]
    with $\delta'=\delta/(AX+1)$ and a fixed $\gamma_1\in (0,1]$ we will specify later. This yields, with a probability at least $1-\delta'$, that

    \begin{align*}
        \sum_{t=1}^T \pa{\ell_h^t\indic{x=x_h^t,a=a_h^t}-\E\bra{\ell_h^t\indic{x=x_h^t,a=a_h^t}\middle|\salgebra^{t-1}}}&\leq \gamma_1\sum_{t=1}^T\E\bra{\pa{\ell_h^t}^2\indic{x=x_h^t,a=a_h^t} \middle| \salgebra^{t-1}}+\frac{\iota}{\gamma_1}\\
        &\leq \gamma_1\sum_{t=1}^T\E\bra{\ell_h^t\indic{x=x_h^t,a=a_h^t} \middle| \salgebra^{t-1}}+\frac{\iota}{\gamma_1}\,.
    \end{align*}

    By definition of the estimated loss, $\ell_h^t\indic{x=x_h^t,a=a_h^t}/\mu^s_{1:}(x,a)=\hell^t(x,a)$. We thus divide by $\mu^s_{1:}(x,a)$ both sides of the inequality, and the unbiasedness of the loss estimator yields

    \begin{align*}
        \sum_{t=1}^T \bra{\hell^t(x,a)-\ell^t(x,a)}&\leq \gamma_1\sum_{t=1}^T \ell^t(x,a)+\frac{\iota}{\gamma_1\mu^s_1:(x,a)}\,.
    \end{align*}

    This inequality holds for all $(x,a)$ with a probability of at least $1-\delta\AX/(\AX+1)$. Taking the scalar product with any $\mu\in\maxpi$ then gives

    \begin{align*}
        \sum_{t=1}^T\scal{\hell^t-\ell^t}{\mu_{1:}}&\leq \gamma_1\sum_{t=1}^T \scal{\ell^t}{\mu_{1:}}+\frac{1}{\gamma_1}\sum_{x\in\cX}\sum_{a\in\cA(x)}\frac{\mu_{1:}(x,a)}{\mu_{1:}^s(x,a)}\\
        &\leq \gamma_1 HT+\frac{\iota}{\gamma_1}\kappa(\mu^s)\,.
    \end{align*}

    Using $\gamma_1=\sqrt{\iota\kappa(\mu^s)/(HT)}\leq 1$ (by assumption), finally yields

    \[\textrm{EST I}\leq 2\sqrt{\iota H\kappa(\mu^s)T}\,.\] 

\textit{Upper bound of EST II} For this upper bound, we apply Lemma~\ref{lemma_concentration} directly to the sequence $U^t=\scal{-\hell^t}{\mu_{1:}^t}$. We now choose $\gamma_2\in \R_+$ (no further assumption is needed on $\gamma_2$ as the sequence is negative) and apply the lemma to get with probability at least $1-\delta/(\AX+1)$ 

\begin{align*}
    \sum_{t=1}^T\scal{\ell^t-\hell^t}{\mu_{1:}^t}&\leq \gamma_2\sum_{t=1}^T \E\bra{\scal{\hell^t}{\mu_{1:}^t}^2\middle|\salgebra^{t-1}}+\frac{\iota}{\gamma_2}\\
    &=\gamma_2\sum_{t=1}^T\E\bra{\pa{\sum_{h=1}^H\pa{\ell_h^t}\sum_{x\in\cX}\sum_{a\in\cAx}\indic{x=x_h^t,a=a_h^t}\frac{\mu_{1:}^t(x,a)}{\mu_{1:}^s(x,a)}}^2\middle|\salgebra^{t-1}}+\frac{\iota}{\gamma_2}\\
    \textrm{(Cauchy-Schwarz)\quad}&\leq \gamma_2 H\sum_{t=1}^T\E\bra{\sum_{h=1}^H\pa{\ell_h^t}^2\sum_{x\in\cX}\sum_{a\in\cAx}\indic{x=x_h^t,a=a_h^t}\frac{\mu_{1:}^t(x,a)^2}{\mu_{1:}^s(x,a)^2}\middle|\salgebra^{t-1}}+\frac{\iota}{\gamma_2}\\
    &\leq \gamma_2 H\sum_{t=1}^T\E\bra{\sum_{h=1}^H \ell_h^t\sum_{x\in\cX}\sum_{a\in\cAx}\indic{x=x_h^t,a=a_h^t}\frac{\mu_{1:}^t(x,a)}{\mu_{1:}^s(x,a)^2}\middle|\salgebra^{t-1}}+\frac{\iota}{\gamma_2}\\
    &=\gamma_2 H\sum_{t=1}^T\E\bra{\sum_{h=1}^H \sum_{x\in\cX}\sum_{a\in\cAx}\hell^t(x,a)\frac{\mu_{1:}^t(x,a)}{\mu_{1:}^s(x,a)}\middle|\salgebra^{t-1}}+\frac{\iota}{\gamma_2}\\
    &= \gamma_2 H \sum_{t=1}^T\sum_{x\in\cX}\sum_{a\in\cAx}\ell^t(x,a)\frac{\mu_{1:}^t(x,a)}{\mu_{1:}^s(x,a)}+\frac{\iota}{\gamma_2}\\
    (\textrm{as\;} \ell^t(x,a)\leq 1)\quad &\leq \gamma_2 H \sum_{t=1}^T\sum_{x\in\cX}\sum_{a\in\cAx}\frac{\mu_{1:}^t(x,a)}{\mu_{1:}^s(x,a)}+\frac{\iota}{\gamma_2}\\
    &\leq \gamma_2 H\kappa(\mu^s)T +\frac{\iota}{\gamma_2}\,.
\end{align*}

Taking $\gamma_2=\sqrt{\frac{\iota}{H\kappa(\mu^s)T}}$ then leads to 
\[\sum_{t=1}^T\scal{\ell^t-\hell^t}{\mu_{1:}^t}\leq 2\sqrt{\iota H\kappa(\mu^s)T}\,.\]

Summing the two inequalities yields the inequality of the theorem with a probability of at least $1-\delta$.
\end{proof}

\section{Balanced policy and $\kappa$} \label{app:kappa}

This section deals with the $\kappa(\mu^s)$ and local $\kappa(\mu^s|x)$ of the main paper, and links it to the balanced policy $\mu^\star$.

\paragraph{Recursive $\kappa$ computation} Let $\mu^s$ be the positive sample policy. For any $\mu\in\maxpi$ and $x\in\cX$ of depth $h$, we define $\kappa_\mu(\mu^s|x)$ the local sum of ratios against $\mu$ in the subtree induced by $x$, i.e. 
\[\kappa_\mu(\mu^s|x):=\sum_{x'\in\cX, x\textrm{ is in the history of }x'}\;\sum_{a'\in\cA(x')}\frac{\mu_{h:}(x',a')}{\mu^s_{h:}(x',a')}\]
where, if $(x'_1,a'_1...,x'_{h'},a')$ is the history of $(x',a')$, 
\[\mu_{h:}(x',a'):=\Pi_{i=h}^{h'}\ \mu(a'_i | x'_i)\, .\]
We then formally define $\kappa(\mu^s |x)$ as $\kappa(\mu^s |x):=\max_{\mu\in\maxpi} \kappa_\mu(\mu^s |x)$. For any $\mu\in\maxpi$, the following recursive formula stands
\[\kappa_\mu(\mu^s|x)=\sum_{a\in\cA(x)}\frac{\mu(a|x)}{\mu^s(a|x)}\pa{1+\sum_{x'\in\cX, x'\textrm{ directly follows }(x,a)}\kappa_{\mu}(\mu^s|x')} \]
that follows from the definition of $\kappa_\mu(\mu^s|x)$. The same kind of recursion can then be obtained for $\kappa(\mu^s|x)$, because each appearance of $\mu$ in the previous equality can be maximized independently (depending on different information sets). This yields

\begin{align}
    \kappa(\mu^s|x)&=\max_{\mu\in\Delta_{\cA(x)}}\sum_{a\in\cA(x)}\frac{\mu(a)}{\mu^s(a|x)}\pa{1+\sum_{x'\in\cX, x'\textrm{ directly follows }(x,a)}\kappa(\mu^s|x')}\nonumber \\
    &=\max_{a\in\cA(x)}\frac{1}{\mu^s(a|x)}\pa{1+\sum_{x'\in\cX, x'\textrm{ directly follows }(x,a)}\kappa(\mu^s|x')}\, \label{rec_kappa},
\end{align}
which allows for a simple recursive computation of $\kappa(\mu^s | x)$. Finally, once the whole recursive computation is done, $\kappa(\mu^s)$ itself can be computed by, defining $\cX_1$ the information sets of depth $1$,
\[\kappa(\mu^s)=\sum_{x_1\in\cX_1}\kappa(\mu^s|x_1)\, .\]

\paragraph{Balanced policy} $\kappa(\mu^s|x)$ can also be minimized over $\mu^s\in\maxpi$ recursively from the leaves using the tree structure. Indeed, for each $x\in\cX$, assuming that the minimizers of $\kappa(\mu^s | x')$ are already known for subsequent $x'$, the policy $\mu^s\in\Delta_{\cA(x)}$ that minimizes the maximum along the actions $a\in\cA(x)$ can be computed from \eqref{rec_kappa}. Furthermore, if we define $A^\tau(x,a)$ and $A^\tau(x)$ the total number of actions in the subtrees respectively induced by $(x,a)$ and $x$, i.e. 
\[A^\tau(x,a):=1+\sum_{x'\in\cX, (x,a)\textrm{ is in the history of }x'}\abs{\cA(x')}\quad\textrm{and}\quad A^\tau(x):=\sum_{a\in\cA(x)}A^\tau(x,a)\, ,\]
we can show that $\min_{\mu^s\in\maxpi} \kappa(\mu^s | x)=A^\tau(x)$, and that the minimum is attained by the balanced policy $\mu^\star$ defined by
\[\mu^\star(a|x):=\frac{A^\tau(x,a)}{A^\tau(x)}\, .\]
Indeed, if we assume in \eqref{rec_kappa} that the previous property holds for the $\kappa(\mu^s|x')$, then

\[\kappa(\mu^s|x)=\max_{a\in\cA(x)}\frac{1}{\mu^s(a|x)}\pa{1+\sum_{x'\in\cX, x'\textrm{ directly follows }(x,a)}A^\tau(x')}\\
    =\max_{a\in\cA(x)}\frac{A^\tau(x,a)}{\mu^s(a|x)}\]

and the previous equality is minimized when the $\mu^s(a|x)$ are proportional to the $A^\tau(x,a)$, achieved by the balanced policy $\mu^\star$. With this policy, the same equality gives $\kappa(\mu^\star|x)=A^\tau(x)$, which concludes the induction. 

Finally, computing $\kappa(\mu^\star)$ yields
\[\kappa(\mu^\star)=\sum_{x_1\in\cX_1}\kappa(\mu^\star | x_1)=\sum_{x_1\in\cX_1}A^\tau(x_1)=\AX\,.\]

\section{Generalized dual stabilized online mirror descent} \label{app:gds}

This section will establish the bound related to the updates \eqref{eq:gsmd} obtained with any Legendre function.

\subsection{General Bregman divergence properties}\label{app:breg_properties}

We start this section by stating multiple properties of the Bregman divergence $\breg_{\Psi}$ for $\Psi$ a convex function, continuously differentiable on an open $\Omega$ and defined on $\overline{\Omega}$, that can be found in \citep{cesa-bianchi2006prediction}.

\textit{Law of cosines :} For any $x\in\overline{\Omega}$ and $w,z\in\Omega$, the following equality holds
\[\breg_\Psi(x,w)=\breg_\Psi(x,z)+\breg_\Psi(z,w)-\scal{\nabla\Psi(w)-\nabla\Psi(z)}{x-z}\, .\]

\textit{Bregman projection :} For $\mathcal{C}$ a closed convex of $\overline{\Omega}$, and $\Psi$ strictly convex, we can define the Bregman projection $\Pi_\mathcal{C}^\Psi$ over $\overline{\Omega}$ by
\[\Pi_\mathcal{C}^\Psi(w)=\argmin_{z\in\mathcal{C}}\breg_{\Psi}(z,w)\, .\]
This Bregman projection satisfies a generalized Pythagorean inequality, for $w\in\Omega$ and $z\in\mathcal{C}$
\[\breg_\Psi(z,w)\geq \breg_\Psi(z,\Pi_\mathcal{C}^\Psi(w))+\breg_\Psi(\Pi_\mathcal{C}^\Psi(w),w)\]

\textit{Fenchel dual :} We defined the Fenchel dual $\Psi^\star$ of a Legendre function $\Psi$ for any $\xi\in\R^n$ by
\[\Psi^\star(\xi)=\sup_{w\in \overline{\Omega}}\scal{\xi}{w}-\Psi(w)\, .\]
If we consider $\Omega^\star:=\nabla\Psi(\Omega)$, it can be shown that $\nabla\Psi^\star$ is the inverse function of $\nabla\Psi$ over $\Omega^\star$, i.e. for any $w\in\Omega$, $\nabla\Psi^\star(\nabla\Psi(w))=w$. Furthermore, for $w,z\in\Omega$,
\[\breg_{\Psi}(w,z)=\breg_{\Psi^\star}(\nabla\Psi^\star(z),\nabla\Psi^\star(y))\, .\]

\textit{Strong convexity:} $\Psi$ is said to be $1$-strongly convex with respect to a norm $\norm{\cdot}$ if for all $w,z\in\Omega$
\[\Psi(z)\geq \Psi(w)+\scal{\nabla\Psi(w)}{z-w}+\frac{1}{2}\norm{w-z}^2\, .\]
In this case, the Bregman divergence of the Fenchel dual $\Psi^\star$ satisfies for any $\xi_1,\xi_2\in\Omega^\star$
\[\breg_{\Psi^\star}(\xi_1,\xi_2)\leq \norm{\xi_1-\xi_2}^{2}_\star\]
where $\norm{\cdot}_\star$ is the dual norm of $\norm{\cdot}$.

\subsection{GDS-OMD Analysis}

We will assume in the following parts that the updates of the following algorithm are properly defined, which happens when all vectors $y^{t+1}$ belong to the Fenchel dual space $\Omega^{t+1,\star}:=\nabla\Psi^{t+1}(\Omega)$. We make the same assumption on the regular OMD iterates $z^t-\xi^t$.

\begin{algorithm}[t] 
\caption{Generalized dual-stabilized online mirror descent}
\label{alg_general_entropy}
\begin{algorithmic}[1]
			\STATE \textbf{Input:}\\
			~~~~ A sequence of dual iterates $\xi^t$\\
                ~~~~ An open subset $\Omega\in\R^n$ and a closed convex $\mathcal{C}$ of $\overline{\Omega}$\\
                ~~~~ A sequence of Legendre regularizers $(\Psi^t)_{t\in [T]}$ on $\overline{\Omega}$ such that for all $t\in[T]$, $\Psi^{t+1}-\Psi^t$ is convex\\
                ~~~~ An initial primal iterate $w^1\in \mathcal{C}$\\

                \STATE \textbf{Output:}\\
                ~~~~ A sequence $(w^t)_{t\in[T]}$ of primal iterates
                \STATE \textbf{Algorithm:}\\

                For $t=1$ to $T$\\
                ~~~~ $z^t=\nabla\Psi^t(w^t)$\\
                ~~~~ $y^{t+1}=z^t-\xi^t+\nabla\Psi^{t+1}(w_1)-\nabla\Psi^t(w^1)$\\
                ~~~~ $\hat{w}^{t+1}=\nabla\Psi^{t+1,\star}(y^{t+1})$\\
                ~~~~ $w^{t+1}=\Pi_\mathcal{C}^{\Psi^{t+1}}(\hat{w}^{t+1})$
\end{algorithmic}
\end{algorithm}

We start by giving an equivalent formulation of the updates \eqref{eq:gsmd} through Algorithm~\ref{alg_general_entropy}.
\begin{proposition}
    Algorithm~\ref{alg_general_entropy} computes the updates \eqref{eq:gsmd} if they are properly defined, i.e. computes the sequence of primal iterates defined by
    \[w^{t+1}=\argmin_{w\in\mathcal{C}} \scal{\xi^t}{w}+\breg_{\Psi^t}\pa{w,w^t}+\breg_{\Psi^{t+1}-\Psi^t}\pa{w,w^1}\, .\]
\end{proposition}

\begin{proof}
    By definition of $\hat{w}^{t+1}$ in Algorithm~\ref{alg_general_entropy}, we have for all iterations $t\in[T]$ and $w\in\mathcal{C}$
    \begin{align*}
        \breg_{\Psi^{t+1}}(w,\hat{w}^{t+1})&=\Psi^{t+1}(w)-\scal{\nabla\Psi^{t+1}(\hat{w}^{t+1})}{w}+C_1\\
        &=\Psi^t(w)+\pa{\Psi^{t+1}(w)-\Psi^t(w)}-\scal{y^{t+1}}{w}+C_1\\
        &=\scal{\xi^t}{w}+\pa{\Psi^t(w)-\scal{\nabla\Psi^{t}(w^t)}{w}}+\\
        &\qquad\qquad\pa{\Psi^{t+1}(w)-\Psi^t(w)-\scal{\nabla\Psi^{t+1}(w^1)-\nabla\Psi^t(w^1)}{w}}+C_1\\
        &=\scal{\xi^t}{w}+\breg_{\Psi^t}\pa{w,w^t}+\breg_{\Psi^{t+1}-\Psi^t}\pa{w,w^1}+C_2
    \end{align*}
    where $C_1$ and $C_2$ are constants independent of the choice of $w$ (but not independent of the other variables). As $w^{t+1}=\argmin_{w\in\mathcal{C}} \breg_{\Psi^{t+1}}(w,\hat{w}^{t+1})$, the updates of Algorithm~\ref{alg_general_entropy} coincide with the updates \eqref{eq:gsmd}, as both minimize the same function at each iteration up to an additive constant.
\end{proof}

The updates of Algorithm~\ref{alg_general_entropy} are then used to show Theorem~\ref{thm:general_entropy} below. Compared to the ones of \citep{McMahansurvey} that also allow adaptive regularization, these updates do not suffer from the potential linear rates observed in \citep{orabona2018scale}.

\begin{theorem}\label{thm:general_entropy}
Let $(w^t)_{t\in[T]}$ be a sequence of primal iterates generated by the updates \eqref{eq:gsmd}, with convex incremental functions. Then for any $w^\dagger\in \overline{\Omega}$,
\[\sum_{t=1}^T\scal{\xi^t}{w^t-w^\dagger}\leq \vphantom{\sum_{t=1}^T}\breg_{\Psi^{T}}(w^\dagger,w^1)+\sum_{t=1}^T\breg_{\Psi^{t,\star}}\pa{\nabla\Psi^t(w^t)-\xi^t,\nabla\Psi^t(w^t)}\]
\end{theorem}

\begin{proof}
    We can assume, without any incidence on the $(w^t)_{t\in[T]}$ sequence, that $\Psi^{T+1}=\Psi^T$. We also define for all $t\in[T]$ the notations $\varphi^t=\Psi^{t+1}-\Psi^t$ and
    \[\hat{q}^t=\scal{\xi^t}{\hat{w}^{t+1}}+\breg_{\Psi^t}(\hat{w}^{t+1},w^t)+\breg_{\varphi^t}(\hat{w}^{t+1},w^1)\, .\] 
    We then divide the sum into a stability and a penalty terms:
    \[\sum_{t=1}^T\scal{\xi^t}{w^t-w^\dagger}= \underbrace{\sum_{t=1}^T\pa{\hat{q}^t-\scal{\xi^t}{w^\dagger}}}_{\mathrm{penalty}}+\underbrace{\sum_{t=1}^T\pa{\scal{\xi^t}{w^t}-\hat{q}^t}}_{\mathrm{stability}}\]
    and we look at upper-bounding these two terms.
    
    \textit{Penalty term}: For all $t\in[T]$, using the law of cosines on the Bregman divergences of $\Psi^t$ and $\varphi^t$, we have the two equalities:
    \[\breg_{\Psi^t}(w^\dagger,w^t)=\breg_{\Psi^t}(w^\dagger,\hat{w}^{t+1})+\breg_{\Psi^t}(\hat{w}^{t+1},w^t)-\scal{\nabla\Psi^t(w^t)-\nabla\Psi^t(\hat{w}^{t+1})}{w^\dagger-\hat{w}^{t+1}}\]
    and
    \[\breg_{\varphi^t}(w^\dagger,w^1)=\breg_{\varphi^t}(w^\dagger,\hat{w}^{t+1})+\breg_{\varphi^t}(\hat{w}^{t+1},w^1)-\scal{\nabla\varphi^t(w^1)-\nabla\varphi^t(\hat{w}^{t+1})}{w^\dagger-\hat{w}^{t+1}}\,.\] 
    Summing these two equalities, we get
    \begin{align*}
        \breg_{\Psi^t}&(w^\dagger,w^t)+\breg_{\varphi^t}(w^\dagger,w^1)\\
        &=\breg_{\Psi^{t+1}}(w^\dagger,\hat{w}^{t+1})+\breg_{\Psi^t}(\hat{w}^{t+1},w^t)+\breg_{\varphi^t}(\hat{w}^{t+1},w^1)-\scal{\xi^t}{w^\dagger-\hat{w}^{t+1}}\\
        &=\breg_{\Psi^{t+1}}(w^\dagger,\hat{w}^{t+1})+\hat{q}^t-\scal{\xi^t}{w^\dagger}
    \end{align*}
    as by definition of $\hat{w}^{t+1}$ and $y^{t+1}$, 
    \[\nabla\Psi^{t+1}(\hat{w}^{t+1})=y^{t+1}=-\xi_t+\nabla\Psi^t(w^t)+\nabla\varphi^t(w^1)\,.\]
    Furthermore, as $w^{t+1}=\Pi_C^{t+1}(\hat{w}^{t+1})$, the Pythagorean inequality for the Bregman divergence yields that 
    \[\breg_{\Psi^{t+1}}(w^\dagger,\hat{w}^{t+1})\geq\breg_{\Psi^{t+1}}(w^\dagger,w^{t+1})+\breg_{\Psi^{t+1}}(w^{t+1},\hat{w}^{t+1})\geq \breg_{\Psi^{t+1}}(w^\dagger,w^{t+1})\, .\]
    Injecting this in the previous equality and telescoping leads to
    \begin{align*}
        \sum_{t=1}^T\pa{\hat{q}^t-\scal{\xi^t}{w^\dagger}}&=\sum_{t=1}^T\pa{\breg_{\Psi^t}(w^\dagger,w^t)+\breg_{\varphi^t}(w^\dagger,w^1)-\breg_{\Psi^{t+1}}(w^\dagger,\hat{w}^{t+1})}\\
        &\leq \sum_{t=1}^T\pa{\breg_{\Psi^t}(w^\dagger,w^t)+\breg_{\varphi^t}(w^\dagger,w^1)-\breg_{\Psi^{t+1}}(w^\dagger,w^{t+1})}\\
        &= \breg_{\Psi^{T+1}}(w^\dagger,w^1)-\breg_{\Psi^{T+1}}(w^\dagger,w^{t+1})\\
        &\leq \breg_{\Psi^T}(w^\dagger,w^1)
    \end{align*}
    as $\Psi^T=\Psi^{T+1}$ by definition.

    \textit{Stability term}: We first notice, for all $t\in[T]$, that
    \begin{align*}
        \scal{\xi^t}{w^t}-\hat{q}^t &=\scal{\xi^t}{w^t-\hat{w}^{t+1}}-\breg_{\Psi^t}(\hat{w}^{t+1},w^t)-\breg_{\varphi^t}(\hat{w}^{t+1},w^1)\\
        &\leq \scal{\xi^t}{w^t-\hat{w}^{t+1}}-\breg_{\Psi^t}(\hat{w}^{t+1},w^t)\\
        &\leq\scal{\xi^t}{w^t-\tilde{w}^{t+1}}-\breg_{\Psi^t}(\tilde{w}^{t+1},w^t)
    \end{align*}
    where
    \[\tilde{w}^{t+1}:=\argmin_{\tilde{w}\in \Omega}\bra{\scal{\xi^t}{\tilde{w}}+\breg_{\Psi^t}(\tilde{w},w^t)}\]
    is the $\tilde{w}^{t+1}$ iterate that would be obtained using a classical OMD step with $\Psi^t$, without the stabilization. By optimality, it verifies
    \[\nabla\Psi^t(\tilde{w}^{t+1})=\nabla\Psi^t(w^t)-\xi^t\]
    and the law of cosines then yields
    \begin{align*}
        \breg_{\Psi^t}(w^t,w^t)&=\breg_{\Psi^t}(w^t,\tilde{w}^{t+1})+\breg_{\Psi^t}(\tilde{w}^{t+1},w^t)-\scal{\nabla\Psi^t(w^t)-\nabla\Psi^t(\tilde{w}^{t+1})}{w^t-\tilde{w}^{t+1}}\\
        (0)&= \breg_{\Psi^t}(w^t,\tilde{w}^{t+1})+\breg_{\Psi^t}(\tilde{w}^{t+1},w^t)-\scal{\xi^t}{w^t-\tilde{w}^{t+1}}\, .
    \end{align*}
    Plugging this in the first inequality, we directly get
    \[\scal{\xi^t}{w^t}-\hat{q}^t\leq \breg_{\Psi^t}(w^t,\tilde{w}^{t+1})\]
    and we conclude using 
    \begin{align*}
        \breg_{\Psi^t}(w^t,\tilde{w}^{t+1})&=\breg_{\Psi^{t,\star}}(\nabla\Psi^t(\tilde{w}^{t+1}),\nabla\Psi^t(w^t))\\
        &=\breg_{\Psi^{t,\star}}(\nabla\Psi^t(w^t)-\xi^t,\nabla\Psi^t(w^t))\,.
    \end{align*}
\end{proof}

\section{LocalOMD analysis} \label{app:localomd}

This section will focus on the dilated entropy approach to extensive-form games, and especially on the updates

\[\label{eq:gds-omd dilated}\tag{GDS-OMD dilated}\mu^{t+1}=\argmin_{\mu\in\maxpi} \scal{\hell^t}{\mu^{}_{1:}}+\breg^\mathrm{dil}_{\alpha^t}(\mu,\mu^t)+\breg^\mathrm{dil}_{\alpha^{t+1}-\alpha^{t}}(\mu,\mu^1)\,\]
that are used by \LocalOMD.

\subsection{General analysis}

The following proposition shows that each update of this form can be computed recursively starting from the leaves of the tree. It requires for any $t\in [T]$ the vector $q^t$ that satisfies for any $x\in\cX$ of depth $h$
    \[q^t(x)=\min_{\mu\in\maxpi} \scal{\hell^{t,\arrowx}}{\mu^{\arrowx}_{h:}}+\breg^\mathrm{dil,\arrowx}_{\alpha^t}(\mu,\mu^t)+\breg^\mathrm{dil,\arrowx}_{\alpha^{t+1}-\alpha^{t}}(\mu,\mu^1)\]
    where $\arrowx$ means that the quantity is considered on the sub-tree induced by $x$ rather than the full information set tree, and $\mu_{h:}$ is defined in Appendix~\ref{app:kappa}.

\begin{proposition}
    Consider the previous updates \eqref{eq:gds-omd dilated} and the vectors $(q^t)_{t\in[T]}$ above. Both $\mu^{t+1}$ and $q^t$ can be computed recursively starting from the leaves of the tree through
    \[\mu^{t+1}=\argmin_{\mu\in\Delta_{\cAx}}h^t_x(\mu)\quad\textrm{and}\quad q^t(x)=\min_{\mu\in\Delta_{\cAx}}h^t_x(\mu)\]
    where 
        \[h^t_x(\mu)=\scal{\tell^t(x,\cdot)}{\mu}+(1/\alpha^t(x))\,\breg_x(\mu,\mu^t(\cdot|x))+\pa{{1}/{\alpha^{t+1}(x)}-1/\alpha^t(x)}\,\breg_x(\mu,\mu^1(\cdot|x))\]
    and the regularized loss $\tell^t(x,a)$ is defined by
    \[\tell^t(x,a):=\hell^t(x,a)+\sum_{x'\in\cX |x' \textrm{ directly follows } (x,a)}q^t(x')\, .\]
\end{proposition}

\begin{proof}
    First, note that $\mu^{t+1}$ is the unique minimizer associated to each $q^t(x_1)$ for $x_1$ the information set of depth $1$. Indeed, each of the sub-tree induced by the $x_1$ can be considered as an independent problem. The idea will be to recursively minimize the $q^t(x)$, starting from the leaves (i.e. the final information sets $x_H$), and compute $\mu^{t+1}(|x)$ as the associated minimizer at each information set.

    This minimization is done through, at each $x\in\cX$ of depth $h$, with a decomposition of $q^t(x)$. Indeed, separating the induced tree by $x$ between the root and the rest of the tree leads to
    \begin{align*}
        q^t(x)&
        =\argmin_{\mu\in\maxpi} \scal{\hell^t(x,\cdot)}{\mu(\cdot |x)}+(1/\alpha^t(x))\,\breg_x(\mu(\cdot |x),\mu^t(\cdot|x))\\
        &\quad+\pa{{1}/{\alpha^{t+1}(x)}-1/\alpha^t(x)}\,\breg_x(\mu(\cdot |x),\mu^1(\cdot|x))\\
        &\quad +\sum_{a\in\cAx}\mu(a|x)\sum_{x'\in\cX |\textrm{x' directly follows }(x,a)}\bra{\scal{\hell^{t,\arrowxprime}}{\mu^{\arrowxprime}_{h+1:}}+\breg^\mathrm{dil,\arrowxprime}_{\alpha^t}(\mu,\mu^t)+\breg^\mathrm{dil,\arrowxprime}_{\alpha^{t+1}-\alpha^{t}}(\mu,\mu^1)}\\\vspace{.1cm}
        &=\argmin_{\mu\in\Delta_{\cAx}}\scal{\hell^t(x,\cdot)}{\mu}+(1/\alpha^t(x))\,\breg_x(\mu,\mu^t(\cdot|x))+\pa{{1}/{\alpha^{t+1}(x)}-1/\alpha^t(x)}\,\breg_x(\mu,\mu^1(\cdot|x))\\
        &\quad +\sum_{a\in\cAx}\mu(a)\sum_{x'\in\cX |\textrm{x' directly follows }(x,a)} q^t(x')\\\vspace{.1cm}
        &=\argmin_{\mu\in\Delta_{\cAx}}\scal{\tell^t(x,\cdot)}{\mu}+(1/\alpha^t(x))\,\breg_x(\mu,\mu^t(\cdot|x))+\pa{{1}/{\alpha^{t+1}(x)}-1/\alpha^t(x)}\,\breg_x(\mu,\mu^1(\cdot|x)) \\
        &=\argmin_{\mu\in\Delta_{\cAx}} h_x(\mu)
    \end{align*}
    as each minimization on $\mu\in\maxpi$ is done on independent components. This justifies the recursive computation of both $\mu^{t+1}$ and $q^t$.
\end{proof}

This proposition directly provides the proof of correctness of \LocalOMD, for which the regularized losses at time step $t$ are non-null only on the trajectory with 
\[\tell^t(x,a)=\frac{\indic{x=x_h^t,a=a_h^t}}{\mu_{1:}^s(x)}\tell_h^t\, .\]

We now want to upper(bound the regret associated with this sequence $\mu^t$. The following lemma gives a valuable property that links the regularized loss and the estimated loss.

\begin{lemma}\label{lemma:tree_reg}
    For any policy $\mu'\in\maxpi$, we have

    \[\scal{\tell^t}{\mu'_{1:}}-\sum_{x\in\cX}\mu'_{1:}(x)q^t(x)=\scal{\hell^t}{\mu'_{1:}}-\hat{q}^t\]

    where $\hat{q}^t=\min_{\mu\in\maxpi} \scal{\hell^t}{\mu_{1:}}+\cdil_{\alpha^t}(\mu,\mu^t)+\cdil_{\alpha^{t+1}-\alpha^t}(\mu,\mu^1)$

\end{lemma}

\begin{proof}
    By definition of $\tell^t$ we have, for any $\mu\in\maxpi$
    \begin{align*}
        \scal{\tell^t}{\mu'_{1:}}&=\scal{\hell^t}{\mu'_{1:}}+\sum_{x\in\cX}\sum_{a\in\cA_x} \mu'_{1:}(x,a) \sum_{x' | (x,a)\xrightarrow{}x'}q^t({x'})\\
        &=\scal{\hell^t}{\mu'_{1:}}+\sum_{x\in\cX}\sum_{a\in\cA_x} \sum_{x' | (x,a)\xrightarrow{}x'}\mu'_{1:}(x')q^t({x'})\\
        &=\scal{\hell^t}{\mu'_{1:}}+\sum_{x'\in\cX\backslash \cX_1}\mu'_{1:}(x')q^t({x'})\\
        &=\scal{\hell^t}{\mu'_{1:}}+\sum_{x'\in\cX}\mu'_{1:}(x')q^t({x'})-\sum_{x'\in\cX_1}q^t({x'})
    \end{align*}

    in which we identified the components of the second sum as the set of non-initial information sets. We then conclude using $\sum_{x\in\cX_1}q^t({x})=\hat{q}^t$ by definition of the $q^t$ terms.
\end{proof}

This lemma is then used to upper bound the estimated regret of the sequence generated by the updates \eqref{eq:gds-omd dilated}. Indeed, while we could apply Theorem~\ref{thm:general_entropy}, the associated stability term, which depends on the Fenchel dual of the dilated entropy, is not easy to upper bound. Nonetheless, the proof of the following theorem is mostly the same but with a slightly different definition of the stability and penalty terms.

\begin{theorem}\label{thm:regret_dilated}
Let $(\mu^t)_{t\in [T]}$ be the sequence of policies generated by the updates \eqref{eq:gds-omd dilated}. The following bound holds

\[\hat{R}^T\leq \underbrace{\sup_{\mu^\dagger\in\maxpi}\breg^\mathrm{\mathrm{dil}}_{\alpha^T}(\mu_{1:}^\dagger,\mu_{1:}^1)}_{\mathrm{penalty}}+\underbrace{\sum_{t=1}^T\sum_{x\in\cX}\alpha^t(x)\mu_{1:}^t(x)\breg^\star_x\pa{\nabla\Psi_x(\mu_{1:}^t(\cdot | x))-\frac{1}{\alpha^t(x)}{\tell}^t(x,\cdot),\nabla\Psi_x(\mu_{1:}^t(\cdot | x))}}_{\mathrm{stability}}\, .\]
\end{theorem}

\begin{proof}
    The separation between the stability and the penalty terms is the same as in Theorem~\ref{thm:general_entropy}, but with $\hat{q}^t$ (of Lemma~\ref{lemma:tree_reg}) defined after the projection rather than before. This leads to the decomposition
    \[\hat{\cR}^T= \underbrace{\max_{\mu^\dagger\in\maxpi}\sum_{t=1}^T\pa{\hat{q}^t-\scal{\hell^t}{\mu^\dagger_{1:}}}}_{\mathrm{penalty}}+\underbrace{\sum_{t=1}^T\pa{\scal{\hell^t}{\mu_{1:}^t}-\hat{q}^t}}_{\mathrm{stability}}\, .\]

    \textit{Penalty term: }This part is similar to the general theorem. The optimality of $\mu^{t+1}$ leads to, for any $t\in[T]$,
    \[\nabla\Psi^{t+1}(\mu_{1:}^{t+1})=-\hell^t-g^t+\nabla\Psi^t(\mu_{1:}^t)+\nabla\varphi^t(\mu_{1:}^1)\,.\]
    where $g^t\in Q_\mathrm{max}^\perp$ and $\varphi^t=\Psi^{t+1}-\Psi^t$. We use the same two law of cosines as in Theorem~\ref{thm:general_entropy}
    
    \begin{align*}
        \breg_{\Psi^t}(\mu_{1:}^\dagger,\mu_{1:}^t)&=\breg_{\Psi^t}(\mu_{1:}^\dagger,\mu_{1:}^{t+1})+\breg_{\Psi^t}(\mu_{1:}^{t+1},\mu_{1:}^t)-\scal{\nabla\Psi^t(\mu_{1:}^t)-\nabla\Psi^t(\mu_{1:}^{t+1})}{\mu_{1:}^\dagger-\mu_{1:}^{t+1}}\\
        \breg_{\varphi^t}(\mu_{1:}^\dagger,\mu_{1:}^1)&=\breg_{\varphi^t}(\mu_{1:}^\dagger,\mu_{1:}^{t+1})+\breg_{\varphi^t}(\mu_{1:}^{t+1},\mu_{1:}^1)-\scal{\nabla\varphi^t(\mu_{1:}^1)-\nabla\varphi^t(\mu_{1:}^{t+1})}{\mu_{1:}^\dagger-\mu_{1:}^{t+1}}
    \end{align*}
    which yields by summing

    \begin{align*}
        \breg_{\Psi^t}&(\mu_{1:}^\dagger,\mu_{1:}^t)+\breg_{\varphi^t}(\mu_{1:}^\dagger,\mu_{1:}^1)\\
        &=\breg_{\Psi^{t+1}}(\mu_{1:}^\dagger,\mu_{1:}^{t+1})+\breg_{\Psi^t}(\mu_{1:}^{t+1},\mu_{1:}^t)+\breg_{\varphi}(\mu_{1:}^{t+1},\mu_{1:}^1)-\scal{\hell^t+g^t}{\mu_{1:}^\dagger-\mu_{1:}^{t+1}}\\
        &=\breg_{\Psi^{t+1}}(\mu_{1:}^\dagger,\mu_{1:}^{t+1})+\hat{q}^t-\scal{\hell^t}{\mu_{1:}^\dagger}
    \end{align*}

    where we used $\scal{g^t}{\mu_{1:}^\dagger-\mu_{1:}^{t+1}}=0$ from the orthogonality. Summing over $t\in[T]$ then gives, by telescoping similarly to the general theorem,

    \begin{align*}
        \sum_{t=1}^T\pa{\hat{q}^t-\scal{\hell^t}{\mu_{1:}^\dagger}}&=\sum_{t=1}^T\pa{\breg_{\Psi^t}(\mu_{1:}^\dagger,\mu_{1:}^t)+\breg_{\varphi^t}(\mu_{1:}^\dagger,\mu_{1:}^1)-\breg_{\Psi^{t+1}}(\mu_{1:}^\dagger,\mu_{1:}^{t+1})}\\
        &= \breg_{\Psi^{T+1}}(\mu_{1:}^\dagger,\mu_{1:}^1)-\breg_{\Psi^{T+1}}(\mu_{1:}^\dagger,\mu_{1:}^{t+1})\\
        &\leq \breg_{\Psi^T}(\mu_{1:}^\dagger,\mu_{1:}^1)
    \end{align*}

    \textit{Stability term:} From Lemma~\ref{lemma:tree_reg} used with $\mu'=\mu^t$, we get an alternative expression of the stability term
    \[\scal{\hell^t}{\mu^t_{1:}}-\hat{q}^t=\scal{\tell^t}{\mu^t_{1:}}-\sum_{x\in\cX}\mu^t_{1:}(x)q^t(x)\]
    This shows the stability term can be decomposed in a positive linear combination
    \[\scal{\hell^t}{\mu^t_{1:}}-\hat{q}^t=\sum_{x\in\cX} \mu^t_{1:}(x) \bra{\scal{\tell^t(x,\cdot)}{\mu^t(\cdot|x)}-q^t(x)}\]
    and we will individually upperbound each of the terms of the combination. The method is again similar to the general theorem, but locally with the regularized loss. Defining $\Psi^t_x:=\alpha^t(x)\Psi_x$ and $\varphi^t_x:=\Psi^{t+1}_x-\Psi^t_x$, we have
    \begin{align*}
        &\scal{\tell^t(x,\cdot)}{\mu^t(\cdot|x)}-q^t(x) \\
        &\qquad=\scal{\tell^t(x,\cdot)}{\mu^t(\cdot|x)-\mu^{t+1}(\cdot|x)}-\breg_{\Psi_x^t}(\mu^{t+1}(\cdot|x),\mu^t(\cdot|x))-\breg_{\varphi_x^t}(\mu^{t+1}(\cdot|x),\mu^1(\cdot|x))\\
        &\qquad\leq \scal{\tell^t(x,\cdot)}{\mu^t(\cdot|x)-\mu^{t+1}(\cdot|x)}-\breg_{\Psi_x^t}(\mu^{t+1}(\cdot|x),\mu^t(\cdot|x))\\
        &\qquad\leq\scal{\tell^t(x,\cdot)}{\mu^t(\cdot|x)-\tilde{\mu}^{t+1}(\cdot|x)}-\breg_{\Psi_x^t}(\tilde{\mu}^{t+1}(\cdot|x),\mu^t(\cdot|x))
    \end{align*}
    where 
    \[\tilde{\mu}^{t+1}(\cdot|x):=\argmin_{\tilde{\mu}\in \Omega_x}\bra{\scal{\tell^t(x,\cdot)}{\tilde{\mu}}+\breg_{\Psi_x^t}(\tilde{\mu},\mu^t(\cdot|x))}\]
    By optimality, $\tilde{\mu}^{t+1}(\cdot|x)$ verifies
    \[\nabla\Psi_x^t(\tilde{\mu}^{t+1}(\cdot|x))=\nabla\Psi_x^t(\mu^t(\cdot|x))-\tell^t(x,\cdot)\]
    and the law of cosines yields
    \begin{align*}
        0&=\breg_{\Psi_x^t}(\mu^t(\cdot|x),\mu^t(\cdot|x))\\
        &=\breg_{\Psi_x^t}(\mu^t(\cdot|x),\tilde{\mu}^{t+1}(\cdot|x))+\breg_{\Psi_x^t}(\tilde{\mu}^{t+1}(\cdot|x),\mu^t(\cdot|x))-\\
        &\qquad\qquad\qquad\scal{\nabla\Psi_x^t(\mu^t(\cdot|x))-\nabla\Psi_x^t(\tilde{\mu}^{t+1}(\cdot|x))}{\mu^t(\cdot|x)-\tilde{\mu}^{t+1}(\cdot|x)}\\
        &= \breg_{\Psi_x^t}(\mu^t(\cdot|x),\tilde{\mu}^{t+1}(\cdot|x))+\breg_{\Psi_x^t}(\tilde{\mu}^{t+1}(\cdot|x),\mu^t(\cdot|x))-\scal{\tell^t(x,\cdot)}{\mu^t(\cdot|x)-\tilde{\mu}^{t+1}(\cdot|x)}
    \end{align*}
    Plugging this in the first inequality, we directly get
    \[\scal{\tell^t(x,\cdot)}{\mu^t(\cdot|x)}-q^t(x)\leq \breg_{\Psi_x^t}(\mu^t(\cdot|x),\tilde{\mu}^{t+1}(\cdot|x))\]
    and we get the individual upper bounds with
    \begin{align*}
        \breg_{\Psi_x^t}(\mu^t(\cdot|x),\tilde{\mu}^{t+1}(\cdot[x))&=\alpha^t(x)\breg_{\Psi_x}(\mu^t(\cdot|x),\tilde{\mu}^{t+1}(\cdot[x))\\
        &=\alpha^t(x)\breg_{\Psi_x^{\star}}(\nabla\Psi_x(\tilde{\mu}^{t+1}(\cdot[x)),\nabla\Psi_x(\mu^t(\cdot|x)))\\
        &=\alpha^t(x)\breg_{\Psi_x^{\star}}\pa{\nabla\Psi_x(\mu^t(\cdot|x))-\frac{1}{\alpha^t(x)}\tell^t(x,\cdot),\nabla\Psi_x(\mu^t(\cdot|x))}
    \end{align*}
\end{proof}

This upper bound on the estimated regret is then used with the learning rates considered in the main article.

\subsection{Optimal rates analysis}

We first consider the optimal rates of the main paper.
\begin{theorem}\label{thm:optimal_main}
    Using \LocalOMD with $\mu^1$ as the uniform policy, with the learning rates $\eta^t(x)= {\eta}/{\kappa(\mu^s | x)}$ where $\eta=\sqrt{\log(A){\kappa(\mu^s)/}{(3HT)}}$, and with $\Psi_x$ the Shannon entropy $\Psi_x(\mu)=\sum_{a\in\cAx}\mu(a)\log(\mu(a))$, the regret is bounded with a probability at least $1-\delta$ by 
    \[\regret^T_{\mathrm{min}}\leq \pa{4+2\sqrt{3}}\,H^{3/2}\sqrt{\log(A)\iota\kappa(\mu^s) T}\quad\textrm{where}\quad \iota=\log(2(\AX+1)/\delta)\,.\]
\end{theorem}

\begin{proof}
    We apply Theorem~\ref{thm:regret_dilated}, using the relations $\alpha^t(x)=1/(\mu^s_{1:}(x)\eta^t(x))$ and $\indic{x=x_h^t}\tell_h^t=\mu^s_{1:}(x)\tell^t(x,\cdot)$. We again separately bound the penalty and stability terms.

    \textit{Penalty term :} We will denote by $\mathrm{PEN}$ this term defined by
    \[\mathrm{PEN}:=\sup_{\mu^\dagger\in\maxpi}\breg^\mathrm{\mathrm{dil}}_{\alpha^T}(\mu_{1:}^\dagger,\mu_{1:}^1)\, .\]
    By definition of the dilated entropy, we have, using that $\mu^1$ is the uniform policy and that the Bregman divergence of the Shannon entropy is the Kullback-Leibler divergence,
    \begin{align*}
    \mathrm{PEN}&=\sup_{\mu^\dagger\in\maxpi}\sum_{x\in\cX}\frac{\mu^\dagger_{1:}(x)\kappa(\mu^s|x)}{\eta\mu^s_{1:}(x)}\breg_{\Psi}(\mu^\dagger(\cdot|x),\mu^1(\cdot|x))\\
    &=\frac{1}{\eta}\sup_{\mu^\dagger\in\maxpi}\sum_{x\in\cX}\frac{\mu^\dagger_{1:}(x)\kappa(\mu^s|x)}{\mu^s_{1:}(x)}\sum_{a\in\cAx}\mu^\dagger(a|x)\log(\mu^\dagger(a|x)/\mu^1(a|x))\\
    &\leq\frac{1}{\eta}\sup_{\mu^\dagger\in\maxpi}\sum_{x\in\cX}\frac{\mu^\dagger_{1:}(x)\kappa(\mu^s|x)}{\mu^s_{1:}(x)}\sum_{a\in\cAx}\mu^\dagger(a|x)\log(1/\mu^1(a|x))\\
    &\leq \frac{\log(A)}{\eta}\sup_{\mu^\dagger\in\maxpi}\sum_{x\in\cX}\frac{\mu^\dagger_{1:}(x)}{\mu^s_{1:}(x)}\kappa(\mu^s|x)\\
    &=\frac{\log(A)}{\eta}\sup_{\mu^\dagger\in\maxpi}\sum_{h=1}^H\sum_{x\in\cX_h}\frac{\mu^\dagger_{1:}(x)}{\mu^s_{1:}(x)}\sup_{\mu'\in\maxpi}\sum_{x'\in\cX | x \textrm{ is in the history of }x'}\sum_{a'\in\cA(x')}\frac{\mu'_{h:}(x',a')}{\mu^s_{h:}(x',a')}\\
    &\leq \frac{\log(A)}{\eta}\sum_{h=1}^H\sup_{\mu^\dagger\in\maxpi}\sum_{x\in\cX_h}\frac{\mu^\dagger_{1:}(x)}{\mu^s_{1:}(x)}\sup_{\mu'\in\maxpi}\sum_{x'\in\cX | x \textrm{ is in the history of }x'}\sum_{a'\in\cA(x')}\frac{\mu'_{h:}(x',a')}{\mu^s_{h:}(x',a')}\\
    (\textrm{by independance}) &= \frac{\log(A)}{\eta}\sum_{h=1}^H\sup_{\mu^\dagger\in\maxpi}\sum_{x\in\cX_h}\frac{\mu^\dagger_{1:}(x)}{\mu^s_{1:}(x)}\sum_{x'\in\cX | x \textrm{ is in the history of }x'}\sum_{a'\in\cA(x')}\frac{\mu^\dagger_{h:}(x',a')}{\mu^s_{h:}(x',a')}\\
   &= \frac{\log(A)}{\eta}\sum_{h=1}^H\sup_{\mu^\dagger\in\maxpi}\sum_{x\in\cX_h}\sum_{x'\in\cX | x \textrm{ is in the history of }x'}\sum_{a'\in\cA(x')}\frac{\mu^\dagger_{1:}(x',a')}{\mu^s_{1:}(x',a')}\\
    &=\frac{\log(A)}{\eta}\sum_{h=1}^H\sup_{\mu^\dagger\in\maxpi}\sum_{x'\in\cX}\sum_{a'\in\cA(x')}\frac{\mu^\dagger_{1:}(x',a')}{\mu^s_{1:}(x',a')}\\
    &=\frac{\log(A)}{\eta}H\kappa(\mu^s)
    \end{align*}
    where $\cX_h$ is the set of information sets of depth $h$, the two sums being later merged on the basis of perfect recall. We now look at the stability term.

    \textit{Stability term :} We will denote by $\mathrm{STA}$ this term defined by
    \[\mathrm{STA}:=\sum_{t=1}^T\sum_{x\in\cX}\alpha^t(x)\mu_{1:}^t(x)\breg^\star_x\pa{\nabla\Psi_x(\mu_{1:}^t(\cdot | x))-\frac{1}{\alpha^t(x)}{\tell}^t(x,\cdot),\nabla\Psi_x(\mu_{1:}^t(\cdot | x))}\]

    We first look at an upper-bound of $\breg^\star_x\pa{\nabla\Psi_x(\mu_{1:}^t(\cdot | x))-\frac{1}{\alpha^t(x)}{\tell}^t(x,\cdot),\nabla\Psi_x(\mu_{1:}^t(\cdot | x))}$. In order to do so, we upper-bound (in the symmetric matrix sense) the Hessian of $\Psi^\star_x$ on $I:=\left\{\nabla\Psi_x(\mu_{1:}^t(\cdot | x))-\frac{\gamma}{\alpha^t(x)}{\tell}^t(x,\cdot)\middle| \gamma\in[0,1]\right\}$.

    Because $\Psi_x (\mu)= \sum_{a\in\cAx} \mu(a)\log(\mu(a))$ is the Shannon entropy,
    \[\nabla\Psi_x(\mu)(a)=\log(\mu(a))+1 \quad\textrm{and thus} \quad \nabla\Psi^\star_x(y)(a)=\exp(y(a)-1)\]
    and the Hessian of $\Psi_x^\star$ is given by
    \[\nabla^2\Psi^\star(y)=\textrm{Diag}\{\exp(y(a)-1)\}_{a\in\cAx}\, .\]
    In particular, it is upper bounded on $I$ by the matrix $D_\mu$ defined by
    \[D_\mu:=\textrm{Diag}\{\mu(a)\}_{a\in\cAx}\]
    This yields that 
    \begin{align*}
        \breg^\star_x\pa{\nabla\Psi_x(\mu_{1:}^t(\cdot | x))-\frac{1}{\alpha^t(x)}{\tell}^t(x,\cdot),\nabla\Psi_x(\mu_{1:}^t(\cdot | x))}&\leq  \frac{1}{2}\norm{\frac{1}{\alpha^t(x)}{\tell}^t(x,\cdot)}^2_{D_\mu^t(\cdot|x)}\\
        &=\frac{1}{2\alpha^t(x)^2}\sum_{a\in\cAx}\mu^t(a|x)\tell^t(x,a)^2
    \end{align*}
    which leads to
    \begin{align*}
        \mathrm{STA}&\leq \sum_{t=1}^T\sum_{x\in\cX}\frac{\mu^t_{1:}(x)}{2\alpha^t(x)}\sum_{a\in\cAx}\mu^t(a|x)\tell^t(x,a)^2\\
        &=\frac{\eta}{2}\sum_{t=1}^T\sum_{x\in\cX}\indic{x=x_h^t}\frac{\mu^t_{1:}(x)}{\mu^s_{1:}(x)}\frac{1}{\kappa(\mu^s|x)}\sum_{a\in\cA(x)}\indic{a=a_h^t}\mu^t(a|x)\tell_h^t(a)^2\, .
    \end{align*}
    We can first notice from recursively comparing the minimizer $\mu^{t+1}(\cdot|x_h^t)$ with $\mu^t(\cdot|x_h^t)$ that the regularized loss $\tell_h^t(a_h^t)$,  satisfies 
    \[\tell_h^t(a_h^t)\leq \scal{\hell^{t,\arrowx}}{\mu^{t,\arrowx}_{h+1:}}\, ,\]
    re-using the notation at the beginning of the section, because the regularization does not evolve with time. The difficulty is now to upper bound $\mathrm{STA}$ with high probability. In order to do so, we use the Lemma~\ref{lemma_concentration} on the sequence $(U^t)_{t\in[T]}$ defined by
    \[U^t:=\sum_{x\in\cX}\indic{x=x_h^t}\frac{\mu^t_{1:}(x)}{\mu^s_{1:}(x)}\frac{1}{\kappa(\mu^s|x)}\sum_{a\in\cA(x)}\indic{a=a_h^t}\mu^t(a|x)\tell_h^t(a)^2\]
    with $\gamma'=\gamma\in (0,1/(H^2\kappa(\mu^s))]$ and $\delta'=\delta/2$.
    This yields with probability at least $1-\delta/2$
    \[\sum_{t=1}^T U^t \leq\sum_{t=1}^T \E\bra{U^t\middle|\salgebra^{t-1}}+\gamma\sum_{t=1}^T\E\bra{(U^t)^2 \middle| \salgebra^{t-1}}+\iota/\gamma\, .\]

    On the one hand, we have, using $\hell_h^t(a_h^t)\leq \kappa(\mu^s|x)$ and the previous inequality that
    \begin{align*}
        \E\bra{U^t\middle|\salgebra^{t-1}}&\leq \sum_{x\in\cX}p^t(x)\mu^t(x)\sum_{a\in\cAx}\scal{\ell^{t,\arrowx}}{\mu^{t,\arrowx}_{h:}}\\
        &\leq \sum_{x\in\cX}p^t(x)\mu^t(x)H\\
        &\leq H^2\, .
    \end{align*}

    On the other hand, using the same inequality,
    \begin{align*}
        \E\bra{(U^t)^2 \middle| \salgebra^{t-1}}&=\E\bra{\pa{\sum_{x\in\cX}\indic{x=x_h^t}\frac{\mu^t_{1:}(x)}{\mu^s_{1:}(x)}\sum_{a\in\cA(x)}\indic{a=a_h^t}\scal{\hell_h^t(a)}{\mu^t(a|x)}}^2\middle| \salgebra^{t-1}}\\
        &\leq H\E\bra{\sum_{x\in\cX}\indic{x=x_h^t}\frac{\mu^t_{1:}(x)^2}{\mu^s_{1:}(x)^2}\sum_{a\in\cA(x)}\indic{a=a_h^t}\scal{\hell_h^t(a)^2}{\mu^t(a|x)^2}\middle| \salgebra^{t-1}}\\
        &\leq H\kappa(\mu^s)\E\bra{\sum_{x\in\cX}\indic{x=x_h^t}\frac{\mu^t_{1:}(x)}{\mu^s_{1:}(x)}\sum_{a\in\cA(x)}\indic{a=a_h^t}\scal{\hell_h^t(a)}{\mu^t(a|x)}\middle| \salgebra^{t-1}}\\
        &\leq H\kappa(\mu^s)\sum_{x\in\cX}p^t(x)\mu^t(x)\sum_{a\in\cAx}\scal{\ell^{t,\arrowx}}{\mu^{t,\arrowx}_{h:}}\\
        &\leq H\kappa(\mu^s)\sum_{x\in\cX}p^t(x)\mu^t(x) H\\
        &\leq H^3\kappa(\mu^s)\, .
    \end{align*}

    The following upper bound on the stability term thus holds
    \[\mathrm{STA} \leq \eta\pa{H^2T+\gamma H^3\kappa(\mu^s)T+\frac{\iota}{\gamma}}\, .\]
    Taking $\gamma=1/(H^2\kappa(\mu^s))$, we obtain
    \[\mathrm{STA} \leq \eta\pa{2H^2T+H^2\iota \kappa(\mu^s)}\]

    As the bound of the theorem trivially holds if $T< \iota\kappa(\mu^s)$ (the regret being bounded by $T$ anyway), we even have assuming $T\geq \iota\kappa(\mu^s)$
    \[\mathrm{STA} \leq 3\eta H^2 T\, .\]

    \textit{Conclusion: }Combining all the previous bounds, the estimated regret is bounded, with a probability of at least $1-\delta/2$ by
    \[\hat{\regret}^T\leq \frac{\log(A)}{\eta}H\kappa(\mu^s)+3\eta H^2 T\, .\]
    Taking $\eta=\sqrt{\log(A){\kappa(\mu^s)/}{(3HT)}}$, we obtain 
    \[\hat{\regret}^T\leq 2\sqrt{3}\,H^{3/2}\sqrt{\log(A)\iota\kappa(\mu^s) T}\, .\]
    We finally conclude by combining this bound with Theorem~\ref{thm:estimation} for the true regret, using $\delta'=\delta/2$, such that the two inequalities hold with a probability at least $1-\delta$.
\end{proof}

\subsection{Adaptive rates analysis}

We end this appendix by considering the adaptive setting. We will assume that all regularizers $\Psi_x$ are 1-strongly convex with respect to some norms $\norm{\cdot}_x$, and we will define 
\begin{align*}
    C_\Psi:&=\sup_{x\in\cX,\mu\in\Delta_{\Ax}}\breg_x(\mu,\mu^1(\cdot|x))\\
    C^\star_\Psi:&=\sup_{x\in\cX,a\in\cA_x}\norm{\indic{x,a}}^\star_{x}
\end{align*}
where $\mu^1$ is the initial policy considered in the algorithm and $\indic{x,a}$ is the loss vector $\ell(x,\cdot)$ equal to $1$ for $a\in\cAx$ and $0$ for $a'\in\cAx\backslash \left\{a\right\}$. The following theorem is the formal statement of Theorem~\ref{thm:adaptive_main} in the main article. While being quite general, the upper bound is unsurprisingly not as tight as the previous one.

\begin{theorem}\label{thm:formal_adaptive}

    With such regularizers, assume that the learning rates are locally decreasing and let $\lambda_1,\lambda_2\in\R_{>0}$ be two constants such that for all information set $x\in\cX$,
    \[\max_{t\in[T-1]}\bra{\frac{1}{\eta^{t+1}(x)}-\frac{1}{\eta^t(x)}}\leq\lambda_1 \quad\mathrm{and}\quad 1/\eta^T(x)+\sum_{t=1}^T\eta^t(x)\indic{x=x_h^t}\leq \lambda_2\sqrt{T}\]
    
    Then with a probability at least $1-\delta$, the regret of Algorithm~\ref{alg:omd_fixed}  is upper-bounded by
    \[\regret^T_{\mathrm{max}}\leq\bra{2\bra{\pa{1+\lambda_1}C_\Psi C_\Psi^\star\kappa(\mu^s)}^2 \lambda_2\abs{\cX}+4\sqrt{H\kappa(\mu^s)\iota}}\sqrt{T}\]
    where $\iota=\log(({\AX+1})/{\delta})$.
\end{theorem}

The proof of this theorem will be based on the following lemma that bounds the regularized loss using the $\lambda_1$ constant above.

\begin{lemma}\label{lemma:regularized_bound}
    For all $t\in[T]$ and $h\in[H]$, 
    \[\tell_h^t(a_h^t)\leq (1+\lambda_1 C_\Psi)\kappa(\mu^s|x_h^t)\, .\]
\end{lemma}

\begin{proof}
    The proof is done recursively on $h$, starting from the leaves. Indeed, for $h=H$, the property is immediate as $\tell_h^t(a_h^t)\leq 1/\mu^s(a_H^t|x_H^t)\leq \kappa(\mu^s|x_H^t)$. If we assume that the property holds for a depth $h>1$, then 
    \begin{align*}
        q_h^t&=\min_{\mu\in\Delta_{\cA(x_h^t)}}\scal{ \tell_h^t}{\mu}+\frac{1}{\eta^t(x_h^t)}\breg_x\pa{\mu,\mu^t(\cdot | x_h^t)}+\pa{\frac{1}{\eta^{t+1}(x_h^t)}-\frac{1}{\eta^t(x_h^t)}}\breg_x\pa{\mu,\mu^1(\cdot | x_h^t)}\\
        &\leq \scal{ \tell_h^t}{\mu^t(\cdot | x_h^t)}+\pa{\frac{1}{\eta^{t+1}(x_h^t)}-\frac{1}{\eta^t(x_h^t)}}\breg_x\pa{\mu^t(\cdot | x_h^t),\mu^1(\cdot | x_h^t)}\\
        &\leq \tell_h^t(a_h^t)+\lambda_1 C_\Psi\, .
    \end{align*}
    Then 
    \begin{align*}
        \tell_{h-1}^t(a_{h-1}^t)&=(\ell_{h-1}^t+q_h^t)/\mu^s(a_{h-1}^t|x_{h-1}^t)\\
        &\leq (1+\lambda_1 C_\Psi+\tell_h^t(a_h^t))/\mu^s(a_{h-1}^t|x_{h-1}^t)\\
        &\leq (1+\lambda_1 C_\Psi)(1+\kappa(\mu^s|x_h^t))/\mu^s(a_{h-1}^t|x_{h-1}^t)\\
        &\leq (1+\lambda_1 C_\Psi)\kappa(\mu^s|x_{h-1}^t)
    \end{align*}
    which concludes the induction.
    
\end{proof}

\begin{proof}
    We now prove the theorem. We start with the estimated regret, that we decompose between the penalty term and the stability term using theorem \ref{thm:regret_dilated}.
    
    \textit{Penalty term:} The penalty term $\mathrm{PEN}$ is bounded by
    \begin{align*}
        \mathrm{PEN}&\leq \sup_{\mu^\dagger\in\maxpi}\breg^\mathrm{\mathrm{dil}}_{\alpha^T}(\mu_{1:}^\dagger,\mu_{1:}^1)\\
        &\leq \sup_{\mu^\dagger\in\maxpi}\sum_{x\in\cX}\frac{1}{\eta^T(x)}\frac{\mu_{1:}^\dagger(x)}{\mu_{1:}^s(x)}\breg_x(\mu^\dagger(\cdot|x),\mu^1(\cdot|x))\\
        &\leq C_\Psi\lambda_2\sqrt{T} \sup_{\mu^\dagger\in\maxpi}\sum_{x\in\cX}\frac{\mu^\dagger_{1:}(x)}{\mu^s_{1:}(x)}\\
        &\leq C_\Psi\lambda_2\kappa(\mu^s)\sqrt{T}\, .
    \end{align*}

    \textit{Stability term:} For the stability term $\mathrm{STA}$, we rely on Lemma~\ref{lemma:regularized_bound} and the $1$-strong convexity of $\Psi_x$ with respect to $\norm{\cdot}_x$ (see Appendix~\ref{app:breg_properties}) and get

    \begin{align*}  \mathrm{STA}&=\sum_{t=1}^T\sum_{x\in\cX}\alpha^t(x)\mu_{1:}^t(x)\breg^\star_x\pa{\nabla\Psi_x(\mu_{1:}^t(\cdot | x))-\frac{1}{\alpha^t(x)}{\tell}^t(x,\cdot),\nabla\Psi_x(\mu_{1:}^t(\cdot | x))}\\
    &\leq \sum_{t=1}^T\sum_{x\in\cX}\frac{\mu_{1:}^t(x)}{\alpha^t(x)}\norm{\tell^t(x,\cdot)}^{\star^2}_x\\
    &\leq \sum_{t=1}^T\sum_{x\in\cX}\eta^t(x)\indic{x=x_h^t}\frac{\mu_{1:}^t(x)}{\mu_{1:}^s(x)}\norm{\tell_h^t}^{\star^2}_x\\
    &\leq \bra{C_\Psi^{\star}}^2\sum_{t=1}^T\sum_{x\in\cX}\eta^t(x)\indic{x=x_h^t}\frac{\mu_{1:}^t(x)}{\mu_{1:}^s(x)}\pa{\tell_h^t(a_h^t)}^2\\
    &\leq \bra{\pa{1+\lambda_1}C_\Psi C_\Psi^\star}^2\sum_{t=1}^T\sum_{x\in\cX}\eta^t(x)\indic{x=x_h^t}\frac{\mu_{1:}^t(x)}{\mu_{1:}^s(x)}\kappa(\mu^s|x)^2\\
    &\leq \bra{\pa{1+\lambda_1}C_\Psi C_\Psi^\star}^2\kappa(\mu^s)\sum_{t=1}^T\sum_{x\in\cX}\eta^t(x)\indic{x=x_h^t}\frac{1}{\mu_{1:}^s(x)}\kappa(\mu^s|x)\\
    &\leq \bra{\pa{1+\lambda_1}C_\Psi C_\Psi^\star\kappa(\mu^s)}^2\sum_{t=1}^T\sum_{x\in\cX}\eta^t(x)\indic{x=x_h^t}\\&
    \leq \bra{\pa{1+\lambda_1}C_\Psi C_\Psi^\star\kappa(\mu^s)}^2 \lambda_2\abs{\cX}\sqrt{T}\, .
    \end{align*}

    \textit{Conclusion: }By summing these two upper bounds we get
    \begin{align*}
        \hat{\regret}^T&\leq \bra{C_\Psi\lambda_2\kappa(\mu^s)+\bra{\pa{1+\lambda_1}C_\Psi C_\Psi^\star\kappa(\mu^s)}^2 \lambda_2\abs{\cX}}\sqrt{T}\\
        &\leq 2\bra{\pa{1+\lambda_1}C_\Psi C_\Psi^\star\kappa(\mu^s)}^2 \lambda_2\abs{\cX}\sqrt{T}\, .
    \end{align*}

    The bound is finally obtained using the  Theorem~\ref{thm:estimation} that holds with a probability of at least $1-\delta$ and links the estimated regret to the true regret.
    
\end{proof}

\end{document}